\documentclass[11pt,a4paper]{article}
\usepackage{amssymb}
\usepackage{amsmath}
\usepackage{amsfonts}
\usepackage{bm}
\usepackage{bbm}
\usepackage{amsthm}
\usepackage{mathrsfs}
\usepackage{hyperref}
\usepackage{color}
\usepackage[margin=2.41cm]{geometry}
\usepackage[all,cmtip]{xy}
\usepackage[utf8]{inputenc}
\usepackage{graphicx}
\usepackage{varwidth}
\usepackage{comment}
\usepackage{enumitem}

\usepackage{mathtools}

\usepackage{upgreek}
\usepackage{rotating}

\usepackage{tikz}
\usetikzlibrary{shapes.geometric}
\usetikzlibrary{patterns.meta}
\usetikzlibrary{cd}
\usetikzlibrary{arrows}

\definecolor{darkred}{rgb}{0.8,0.1,0.1}
\hypersetup{
	colorlinks=true,         
	urlcolor=darkred,
	linkcolor=darkred,
	citecolor=blue,
}

\theoremstyle{plain}
\newtheorem{theo}{Theorem}[section]

\newtheorem{propo}[theo]{Proposition}

\theoremstyle{definition}
\newtheorem{defi}[theo]{Definition}

\newtheorem{notation}[theo]{Notation}

\newtheorem{exx}[theo]{Example}
\newenvironment{ex}
{\pushQED{\qed}\exx}
{\popQED\endexx}

\newtheorem{remm}[theo]{Remark}
\newenvironment{rem}
{\pushQED{\qed}\remm}
{\popQED\endremm}

\newtheorem{constrr}[theo]{Construction}
\newenvironment{constr}
{\pushQED{\qed}\constrr}
{\popQED\endconstrr}

\numberwithin{equation}{section}

\def\nn{\nonumber}

\def\bbR{\mathbb{R}}
\def\bbC{\mathbb{C}}

\def\bbZ{\mathbb{Z}}

\def\id{\mathrm{id}}

\def\deg{\mathrm{deg}}
\def\dd{\mathrm{d}}

\def\cc{\mathrm{c}}

\def\dim{\mathrm{dim}}
\def\1{I}
\def\oone{\mathbbm{1}}

\def\A{\mathcal{A}}

\def\E{\mathcal{E}}
\def\F{\mathcal{F}}

\def\I{\mathcal{I}}
\def\L{\mathcal{L}}
\def\M{\mathcal{M}}
\def\O{\mathcal{O}}

\newcommand\ovr[1]{\overline{#1}}

\newcommand\wedgebracket[2]{\left[ #1 {}\stackrel{\scalebox{0.6}{$\wedge$}}{,}{} #2\right]}
\newcommand\wedgepair[2]{\left\langle #1 {}\stackrel{\scalebox{0.6}{$\wedge$}}{,}{} #2\right\rangle}
\newcommand{\pair}[2]{\langle\!\langle #1 , #2 \rangle\!\rangle}

\def\g{\mathfrak{g}}

\def\CP{\mathbb{C}P^1}
\def\delbar{\overline{\partial}}

\def\sk{\vspace{2mm}}

\makeatletter
\let\@fnsymbol\@alph
\makeatother

\title{%
The homological algebra of 2d integrable field theories
}

\author{%
Marco Benini$^{1,2,a}$, Alexander Schenkel$^{3,b}$\ and\ Beno\^{\i}t Vicedo$^{4,c}$\vspace{4mm}\\
{\small ${}^1$ Dipartimento di Matematica, Dipartimento di Eccellenza 2023-27, Universit\`a di Genova,}\\
{\small Via Dodecaneso 35, 16146 Genova, Italy.}\vspace{2mm}\\
{\small ${}^2$ INFN, Sezione di Genova,}\\
{\small Via Dodecaneso 33, 16146 Genova, Italy.}\vspace{2mm}\\
{\small ${}^3$~Dipartimento di Matematica, Universit{\`a} di Trento,}\\
{\small Via Sommarive 14, 38123 Povo (Trento), Italy.}\vspace{2mm}\\
{\small ${}^4$ Department of Mathematics, University of York,}\\
{\small Heslington, York YO10 5GH, United Kingdom.}\vspace{4mm}\\
{\small 
Email: ${}^a$~\href{mailto:marco.benini@unige.it}{\texttt{marco.benini@unige.it}}, 
${}^b$~\href{mailto:alexander.schenkel@unitn.it}{\texttt{alexander.schenkel@unitn.it}},
${}^c$~\href{mailto:benoit.vicedo@gmail.com}{\texttt{benoit.vicedo@gmail.com}}
\vspace{2mm}
}
}

\date{January 2026}

\begin{document}

\maketitle

\begin{abstract}
\noindent This article provides a detailed and rigorous study of $4d$ semi-holomorphic Chern-Simons theories and their associated $2d$ integrable field theories from the homological perspective of $L_\infty$-algebras. Through the use of homotopy transfer techniques, it is shown precisely how both the integrable field theory and its corresponding Lax connection emerge from the $4d$ theory, which results in a novel perspective on Lax connections in terms of $L_\infty$-morphisms. 
\end{abstract}
\vspace{-1mm}

\vspace{-2mm}

\vspace{-2mm}

\renewcommand{\baselinestretch}{0.8}\normalsize
\tableofcontents
\renewcommand{\baselinestretch}{1.0}\normalsize


\section{\label{sec:intro}Introduction and summary}
The systematic construction of integrable field theories 
in two spacetime dimensions has witnessed a major breakthrough
since the seminal paper \cite{CY3} of Costello and Yamazaki.
One of the key insights of this work is that a large class
of integrable field theories on a $2d$ spacetime $\Sigma$ emerges
very naturally from a semi-holomorphic Chern-Simons theory
$S(\A)=\int_X \omega \wedge \mathsf{CS}(\A)$ 
on the $4$-dimensional product manifold $X= \Sigma \times C$,
where $C$ is a Riemann surface encoding the spectral parameters
and $\omega$ is a choice of meromorphic $1$-form on $C$. 
A particularly pleasant feature of this framework is that
Lax connections, which traditionally were rather mysterious objects 
and often constructed only by sophisticated guesswork, are now
represented more directly through the Chern-Simons gauge field $\A$.
The vast flexibility of this approach arises from the fact that
different input data for the $4d$ theory, given concretely by the choice 
of 1.)~a structure group $G$, 2.)~a Riemann surface $C$, 
3.)~a meromorphic $1$-form $\omega$ on $C$ and 4.)~boundary conditions and singularities
for the gauge fields $\A$ at the poles and zeros of $\omega$, yield (a priori) 
different $2d$ integrable field theories on $\Sigma$, see also \cite{DLMV,BSV}
for more details.
\sk

The main goal of our present paper is to provide a conceptual explanation 
and mathematically rigorous realization of the slogan that `$4d$ semi-holomorphic
Chern-Simons theory produces $2d$ integrable field theories together with their corresponding 
Lax connections'. Our results are most naturally stated and proven in a homological
context where field theories are described by (cyclic) $L_\infty$-algebras,
which in physics is known under the name Batalin-Vilkovisky formalism, 
see e.g.\ \cite{CG1,CG2} for a modern perspective and also \cite{JRSW} for a review.
In addition to providing us with powerful tools to formulate concepts and
prove results about classical field theories, such homological framework
will be highly beneficial for any future analysis of $4d$ semi-holomorphic 
Chern-Simons theory at the quantum level.
\sk

To highlight the main innovations of our paper, let us now
present our key constructions and results from a non-technical bird's-eye view 
in the simplest case where $C=\CP$ is the Riemann sphere. (The generalization to 
higher genus Riemann surfaces is possible but slightly more subtle, see Section \ref{sec:highergenus}.)
Our starting point is the auxiliary $L_\infty$-algebra
$\big(\E(X),\ell\big)$ from Definition \ref{def:auxiliaryLinfty}
which models the space of fields, gauge symmetries and equations
of motion of $4d$ semi-holomorphic Chern-Simons theory without
singularities and boundary conditions. This $L_\infty$-algebra
is built from the de Rham complex $\Omega^\bullet(\Sigma)$
on the $2d$ spacetime $\Sigma$ and the Dolbeault complex $\Omega^{0,\bullet}(C)$
on the Riemann surface $C$, giving it the desired topological-holomorphic
features of $4d$ semi-holomorphic Chern-Simons theory, see Remark \ref{rem:auxiliaryLinfty}.
Singularities and boundary conditions of the fields at the zeros and poles of $\omega$
are then modeled mathematically by suitable choices of \textit{divisors} on $C$, which we use to twist
the component-wise Dolbeault complexes entering $\big(\E(X),\ell\big)$. This yields
the $L_\infty$-algebra $\big(\F(X),\ell\big)$ from Definition \ref{def:boundaryLinfty}
which describes $4d$ semi-holomorphic Chern-Simons theory with the chosen singularities
and boundary conditions. Note that this $L_\infty$-algebra provides in particular a rigorous definition of the spaces
of fields and gauge transformations of $4d$ semi-holomorphic Chern-Simons theory 
in the presence of singularities and boundary conditions.
We also show in Subsection \ref{subsec:cyclic} that the $L_\infty$-algebra $\big(\F(X),\ell\big)$
carries a cyclic structure, which provides in particular a precise definition of a gauge invariant
action functional $S(\A)=\int_X \omega \wedge \mathsf{CS}(\A)$ in the presence of singularities and boundary conditions.
\sk

According to the philosophy of Costello and Yamazaki \cite{CY3}, the 
$L_\infty$-algebra $\big(\F(X),\ell\big)$ of $4d$ semi-holomorphic 
Chern-Simons theory on $X$ in the presence of singularities and boundary conditions 
should model a $2d$ integrable field theory on $\Sigma$. The passage
from $X=\Sigma\times C$ to $\Sigma$ is usually described by some process
of `integrating out the fields on the Riemann surface $C$'. In our work
we provide a mathematical formalization of this process in terms of 
homotopy transfer of $L_\infty$-algebras.
The punchline is that, by computing the component-wise divisor-twisted $\delbar$-cohomologies on $C$,
one obtains a weakly equivalent model
\begin{flalign}
\xymatrix@C=4em{
\big(\F(\Sigma),\ell^\prime\big) \ar@{~>}[r]^-{\sim}~&~ \big(\F(X),\ell\big)
}
\end{flalign}
given by an $L_\infty$-algebra $\big(\F(\Sigma),\ell^\prime\big)$
which describes a $\sigma$-model-type field theory on spacetime $\Sigma$, see Proposition
\ref{prop:transferboundary} and Remark \ref{rem:transferboundary}.
According to our knowledge, this is the first complete proof that
$4d$ semi-holomorphic Chern-Simons theory on $X$ with singularities
and boundary conditions is weakly equivalent to a field theory defined on $\Sigma$.
We also show in Subsection \ref{subsec:cyclic} that the cyclic structure on $\big(\F(X),\ell\big)$
transfers to the $L_\infty$-algebra $\big(\F(\Sigma),\ell^\prime\big)$,
which provides in particular an action functional for this $\sigma$-model-type field theory on $\Sigma$.
\sk

Our framework leads further to a very clear and precise identification of the
structure that renders the $\sigma$-model-type field theory $\big(\F(\Sigma),\ell^\prime\big)$ 
integrable. The key observation is that another suitable choice of divisors 
on $C$ yields the $L_\infty$-algebra
$\big(\L(X),\ell\big)$ from Definition \ref{def:singularLinfty} which
describes $4d$ semi-holomorphic Chern-Simons theory with only 
singularities but no boundary conditions. Computing again 
the component-wise divisor-twisted $\delbar$-cohomologies on $C$
and using homotopy transfer, one obtains a weakly equivalent model
\begin{flalign}
\xymatrix@C=4em{
\big(\L(\Sigma),\ell^\prime\big) \ar@{~>}[r]^-{\sim}~&~ \big(\L(X),\ell\big)
}
\end{flalign}
given by an $L_\infty$-algebra $\big(\L(\Sigma),\ell^\prime\big)$
which describes Lax connections on $\Sigma$, see Proposition
\ref{prop:transfersingular} and Remark \ref{rem:transfersingular}.
The main result summarizing our constructions
in this paper is then Theorem \ref{theo:lax}, which shows that
there exists a canonical $L_\infty$-morphism
\begin{flalign}\label{eqn:LaxmorphismIntro}
\xymatrix@C=4em{
\big(\F(\Sigma),\ell^\prime\big) \ar@{~>}[r]^-{}~&~\big(\L(\Sigma),\ell^\prime\big) 
}
\end{flalign}
from the $L_\infty$-algebra $\big(\F(\Sigma),\ell^\prime\big)$
describing a $\sigma$-model-type field theory on spacetime $\Sigma$
to the $L_\infty$-algebra $\big(\L(\Sigma),\ell^\prime\big)$
describing Lax connections. In particular,
this $L_\infty$-morphism assigns to every on-shell field on $\Sigma$
its corresponding flat and meromorphic Lax connection, hence it is the additional 
datum of this $L_\infty$-morphism which renders the field theory on $\Sigma$ integrable.
\sk

As a final remark we would like to mention that our homological framework
presented in this paper is clearly not limited to $4d$ semi-holomorphic Chern-Simons
theory, but it can easily be generalized to topological-holomorphic field theories
in any dimension. This provides powerful and systematic tools to explore
the new territory of higher-dimensional semi-holomorphic Chern-Simons theories \cite{SV,CL},
leading to structural insights into their associated higher-dimensional
integrable field theories. This will be the content of the follow-up article 
\cite{BCSV}.
\sk

Let us now briefly outline the contents of this paper.
In Section \ref{sec:prelim} we recall some basic
concepts and results from the theory of Riemann surfaces
and from homological algebra which will be used in the main text.
In Section \ref{sec:genuszero} we present in detail 
our constructions outlined above in the case where 
$C=\CP$ is the Riemann sphere. In Subsection \ref{subsec:setup}
we explain our basic setup and introduce in Definition \ref{def:auxiliaryLinfty}
the auxiliary $L_\infty$-algebra $\big(\E(X),\ell\big)$ which describes 
$4d$ semi-holomorphic Chern-Simons theory without
singularities and boundary conditions. In Subsection \ref{subsec:singular}
we explain how one can introduce singularities for the fields by using
suitable divisors, leading to the $L_\infty$-algebra $\big(\L(X),\ell\big)$ 
from Definition \ref{def:singularLinfty} which describes 
$4d$ semi-holomorphic Chern-Simons theory with only singularities
but no boundary conditions. We then show in Proposition 
\ref{prop:transfersingular} that this $L_\infty$-algebra
is weakly equivalent to an $L_\infty$-algebra $\big(\L(\Sigma),\ell^\prime\big)$
describing Lax connections on $\Sigma$, see also Remark \ref{rem:transfersingular}.
In Subsection \ref{subsec:boundary} we introduce in Definition \ref{def:bdycondition} 
a class of boundary conditions which are local on the Riemann surface $C$
and define the $L_\infty$-algebra $\big(\F(X),\ell\big)$ 
from Definition \ref{def:boundaryLinfty} which describes 
$4d$ semi-holomorphic Chern-Simons theory with singularities and boundary conditions.
The key result is Proposition \ref{prop:transferboundary} which shows
that the latter is weakly equivalent to an
$L_\infty$-algebra $\big(\F(\Sigma),\ell^\prime\big)$ describing 
a $\sigma$-model-type field theory on spacetime $\Sigma$, see also Remark \ref{rem:transferboundary}.
In Subsection \ref{subsec:Lax} we explain how our formalism gives rise
to the canonical $L_\infty$-morphism \eqref{eqn:LaxmorphismIntro}
which assigns to fields on $\Sigma$ their corresponding Lax connections, 
hence providing the additional datum that renders the theory $\big(\F(\Sigma),\ell^\prime\big)$ integrable.
In Subsection \ref{subsec:cyclic} we introduce a cyclic structure
on $\big(\F(X),\ell\big)$ and show in Proposition 
\ref{prop:cyclictransfer} that it transfers to the 
weakly equivalent $L_\infty$-algebra $\big(\F(\Sigma),\ell^\prime\big)$.
These cyclic structures encode in particular action functionals for both the 
$4d$ theory on $X$ and the integrable field theory on $\Sigma$.
In Section \ref{sec:highergenus} we explain how to generalize our results
to higher genus Riemann surfaces.  For this we adopt 
the proposal in \cite{CY3,Derryberry} to replace
the trivial holomorphic principal $G$-bundle $C\times G\to C$
underlying $4d$ semi-holomorphic Chern-Simons theory by a suitable
non-trivial one in order to kill certain undesired cohomologies.
Our results in the higher genus case are broadly 
analogous to the ones for genus $0$ from Section \ref{sec:genuszero},
however the theory of Lax connections in higher genus is slightly more subtle, 
see in particular Remarks \ref{rem:transfersingularhigher} and  \ref{rem:holonomy}.
In Appendix \ref{app:details} we spell out certain technical details and
constructions which will be used in the main text.


\section{\label{sec:prelim}Preliminaries}
In this section we recall some well-known concepts and results
from the theory of Riemann surfaces and from homological algebra
which are needed to state and prove our results.
The most relevant aspects of Riemann surfaces in the context of our work
consist of divisors and their associated holomorphic line bundles,
which are covered in most textbooks on this subject. Our main reference is the comprehensive and 
well-written textbook of Forster \cite{Forster}. Concerning homological
algebra, we require some basic aspects of the theory of $L_\infty$-algebras,
the homotopy transfer theorem and homological perturbation theory.
The details can be found e.g.\ in \cite{HPT,LodayVallette,KraftSchnitzer},
see also \cite{JRSW} for a more physicists friendly presentation.

\subsection{\label{subsec:Riemannsurfaces}Riemann surfaces}
Let us fix any compact Riemann surface $C$. Following the notations in \cite{Forster}, we 
denote by $\O$ the sheaf of holomorphic functions on $C$, 
by $\O^\ast$ the sheaf of non-vanishing holomorphic functions on $C$,
and by $\M$ and $\M^{(1)}$ the sheaves of meromorphic functions 
and $1$-forms on $C$. The structure of the zeros and poles of a meromorphic 
function or $1$-form can be recorded very conveniently by introducing the
concept of a divisor.
\begin{defi}\label{def:divisor}
A \textit{divisor} $D$ on the compact Riemann surface $C$ is a function
\begin{flalign}
D\,:\, C~\longrightarrow~\bbZ
\end{flalign}
with finite support, i.e.\ $D(p)\neq 0$ only for finitely many points $p\in C$.
The \textit{degree} of the divisor is defined by $\mathrm{deg}(D):= \sum_{p\in C} D(p)\in\bbZ$.
For two divisors $D,D^\prime:C\to\bbZ$, we write $D\leq D^\prime$
if $D(p)\leq D^\prime(p)$ for all $p\in C$.
\end{defi}
\begin{ex}\label{ex:divisor}
Given any meromorphic function $f\in \M(U)$ on an open subset $U\subseteq C$
and any point $p\in U$, we define
\begin{flalign}
\mathrm{ord}_p(f)\,:=\,\begin{cases}
0 &,~~\text{if $f$ is holomorphic and non-zero at $p$}~,\\
k &,~~\text{if $f$ has a zero of order $k$ at $p$}~,\\
-k &,~~\text{if $f$ has a pole of order $k$ at $p$}~,\\
\infty &,~~\text{if $f$ vanishes on an open neighborhood of $p$}~.
\end{cases}
\end{flalign}
Then every non-identically vanishing meromorphic function $f\in\M(C)\setminus\{0\}$ 
defines a divisor
\begin{flalign}
(f)\,:\,C~\longrightarrow~\bbZ~,~~p~\longmapsto~\mathrm{ord}_p(f)
\end{flalign}
which encodes the locations and orders of its zeros and poles.
These definitions extend to meromorphic $1$-forms $\omega\in\M^{(1)}(U)$ by setting
$\mathrm{ord}_p(\omega) := \mathrm{ord}_p(f)$, where $\omega = f\,\dd z$ denotes the 
basis expansion in any choice of local holomorphic coordinate $z$ around $p\in U$.
The assignment of divisors to meromorphic functions and $1$-forms satisfies
the following algebraic identities
\begin{flalign}
(f\,g) \,=\, (f)+(g)~~,\quad 
(1/f)\,=\,-(f)~~,\quad
(f\,\omega)\,=\, (f)+(\omega)\quad,
\end{flalign}
for all $f,g\in\M(C)\setminus\{0\}$  and $\omega\in \M^{(1)}(C)\setminus\{0\}$.
\end{ex}

Associated with any divisor $D:C\to \bbZ$ is a subsheaf $\O_D\subseteq \M$ 
which describes meromorphic functions whose zeros and poles are conditioned 
by the divisor according to 
\begin{flalign}\label{eqn:O_D}
\O_D(U)\,:=\,\Big\{f\in \M(U)\,:\, (f) \geq -D\vert_U \Big\}\quad,
\end{flalign}
for all open subsets $U\subseteq C$, where $D\vert_U : U\to\bbZ$ denotes
the restriction of $D$.
Note that for $D(p)>0$ this limits by $D(p)$ the maximal allowed pole order
of $f$ at $p$, while for $D(p)<0$ it forces $f$ to have a zero at $p$
of order at least $-D(p)$.
\sk

For the purpose of our paper, it will be crucial to realize the sheaves $\O_D$ geometrically
in terms of sheaves of holomorphic sections of a holomorphic line bundle $L_D\to C$.
The construction of these holomorphic line bundles associated with divisors 
is explained in \cite[Section 29.11]{Forster}, but for convenience of 
the reader the main aspects are briefly recalled below.
\begin{constr}\label{constr:linebundles}
Let $D:C\to \bbZ$ be a divisor on the compact Riemann surface $C$. Choose any open cover
$\{U_i\subseteq C\}_{i\in\I}$ such that the restrictions $D\vert_{U_i} = (\psi_i)$ can be represented
by meromorphic functions $\psi_i\in\M(U_i)$ that do not vanish identically on any connected component of $U_i$, for all $i\in\I$. (Such 
open cover exists by \cite[Theorem 26.5]{Forster}.) Define the holomorphic line bundle
$L_D\to C$ in terms of the \v{C}ech data (i.e.\ transition functions)
\begin{flalign}
g_{ij} \,:=\, \frac{\psi_i\vert_{U_{ij}}}{\psi_j\vert_{U_{ij}}}\,\in\,\O^\ast(U_{ij})\quad,
\end{flalign}
which satisfies the required cocycle conditions $g_{ij}\vert_{U_{ijk}}\, g_{jk}\vert_{U_{ijk}}=g_{ik}\vert_{U_{ijk}}$
on all triple overlaps $U_{ijk}$. Up to isomorphism of holomorphic line bundles, this construction
does not depend on the choice of open cover $\{U_i\subseteq C\}_{i\in\I}$ and representatives
$\psi_i\in\M(U_i)$. The sheaf of smooth sections $\Gamma^\infty(L_D)$ of the holomorphic line bundle $L_D\to C$ is given by
\begin{flalign}
\Gamma^\infty(U,L_D)\,:=\,\bigg\{(s_i)_{i\in\I}^{}\in \prod_{i\in\I} C^\infty(U\cap U_i)\,:\, s_i\vert_{U\cap U_{ij}} = 
g_{ij}\vert_{U\cap U_{ij}}\, s_{j}\vert_{U\cap U_{ij}}~~\forall i,j\in\I\bigg\}\quad,
\end{flalign}
for all open subsets $U\subseteq C$. Its sheaf of holomorphic sections $\O(L_D)$ reads as
\begin{flalign}
\O(U,L_D)\,:=\,\Big\{(s_i)_{i\in\I}^{}\in\Gamma^\infty(U,L_D)\,:\,\delbar s_i = 0~~\forall i\in\I\Big\}\quad,
\end{flalign}
for all open subsets $U\subseteq C$.
This sheaf is related to the sheaf $\O_D$ of $D$-conditioned meromorphic functions from \eqref{eqn:O_D} 
through the sheaf isomorphism $\O_D\stackrel{\cong}{\longrightarrow} \O(L_D)$ defined by
\begin{flalign}\label{eqn:O_DGammaisos}
\O_D(U)~\stackrel{\cong}{\longrightarrow}~\O(U,L_D)~~,\quad f~\longmapsto~\big(\psi_i\vert_{U\cap U_i}
\,f\vert_{U\cap U_i}\big)_{i\in \I}^{}\quad,
\end{flalign}
for all open subsets $U\subseteq C$.
\sk

Given two divisors $D,D^\prime:C\to\bbZ$ such that $D\leq D^\prime$,
one can define a holomorphic line bundle
morphism $L_D\to L_{D^\prime}$. Explicitly, choosing any open cover $\{U_i\subseteq C\}_{i\in\I}$
such that $D\vert_{U_i} = (\psi_i)$ and $D^\prime\vert_{U_i} = (\psi^\prime_i)$, for all $i\in\I$,
we define this morphism in terms of the \v{C}ech data
\begin{flalign}
k_i \,:=\, \frac{\psi_i^\prime}{\psi_i}\,\in\, \O(U_i)\quad,
\end{flalign}
which satisfies the required compatibility conditions 
$k_{i}\vert_{U_{ij}}\, g_{ij} = g^\prime_{ij}\, k_{j}\vert_{U_{ij}}$
on all overlaps $U_{ij}$. This yields 
the sheaf morphism $\Gamma^\infty(L_D)\to \Gamma^\infty(L_{D^\prime})$ defined by
\begin{flalign}\label{eqn:linebundlemap}
\Gamma^\infty(U,L_D)~\longrightarrow~\Gamma^\infty(U,L_{D^\prime})~~,\quad
(s_i)_{i\in \I}^{}~\longmapsto~\big(k_i\vert_{U\cap U_i}\,s_i\big)_{i\in\I}^{}\quad,
\end{flalign}
for all open subsets $U\subseteq C$, which restricts to a morphism
$\O(L_D)\to \O(L_{D^\prime})$ between the sheaves of holomorphic sections.
For later use, let us observe that, together with \eqref{eqn:O_DGammaisos} for $D$ and $D^\prime$,
we obtain a commutative diagram of sheaf morphisms
\begin{flalign}\label{eqn:DtoDprimediagram}
\begin{gathered}
\xymatrix{
\ar[d]_-{\cong}\O_D \ar[r]^-{\subseteq}~&~\O_{D^\prime}\ar[d]^-{\cong}\\
\O(L_D)\ar[r]~&~\O(L_{D^\prime})
}
\end{gathered}
\end{flalign}
with the top horizontal arrow the canonical inclusion $\O_{D}\subseteq \O_{D^\prime}$ for 
$D\leq D^\prime$.
\end{constr}

Let us also recall that associated to each Riemann surface $C$ is its Dolbeault complex
\begin{flalign}
\Omega^{0,\bullet}(C)\,:=\,\Big(
\xymatrix{
\Omega^{0,0}(C)\ar[r]^-{\delbar}~&~ \Omega^{0,1}(C)
}
\Big)
\end{flalign}
whose degree $0$ component $\Omega^{0,0}(C) = C^\infty(C)$ is given by the smooth functions on $C$ 
and degree $1$ component is given by the smooth $(0,1)$-forms on $C$,
i.e.\ differential forms taking the form $\xi = f\,\dd \ovr{z}$ in any choice of local holomorphic coordinate $z$.
More generally, given any divisor $D:C\to \bbZ$, one uses its associated holomorphic
line bundle $L_D\to C$ from Construction \ref{constr:linebundles} to 
define the twisted Dolbeault complex
\begin{flalign}
\Omega^{0,\bullet}(C,L_D)\,:=\, 
\Big(
\xymatrix{
\Omega^{0,0}(C,L_D)\ar[r]^-{\delbar}~&~ \Omega^{0,1}(C,L_D)
}
\Big)
\end{flalign}
which describes $(0,\bullet)$-forms on $C$ taking values in 
the holomorphic line bundle $L_D\to C$. By construction, the
zeroth cohomology of this cochain complex 
\begin{flalign}
\mathsf{H}^0\, \Omega^{0,\bullet}(C,L_D)\,=\,\O(C,L_D)\,\cong\,\O_D(C)\,=\,
\mathsf{H}^0(C,\O_D)
\end{flalign}
is isomorphic to the $D$-conditioned meromorphic functions on $C$ from \eqref{eqn:O_D},
which is the zeroth sheaf cohomology of $\O_D$.
Furthermore, as a consequence of Dolbeault's theorem, the first cohomology
\begin{flalign}
\mathsf{H}^1\, \Omega^{0,\bullet}(C,L_D)\,\cong\,\mathsf{H}^1(C,\O_D)
\end{flalign}
is isomorphic to the first sheaf cohomology of $\O_D$. 
These cohomologies are relatively well-understood 
thanks to the powerful Riemann-Roch and Serre duality theorems,
see e.g.\ \cite[Sections 16 and 17]{Forster}. Let us summarize 
the results which are relevant for our work.
\begin{theo}\label{theo:cohomologies}
Let $C$ be a compact Riemann surface of genus $g$ and $D:C\to \bbZ$ a divisor.
\begin{itemize}
\item[(a)] The cohomologies $\mathsf{H}^0 \, \Omega^{0,\bullet}(C,L_D)$ and
$\mathsf{H}^1\, \Omega^{0,\bullet}(C,L_D)$ are finite-dimensional and they
satisfy
\begin{flalign}
\dim\, \mathsf{H}^0\, \Omega^{0,\bullet}(C,L_D) - \dim\,
\mathsf{H}^1\, \Omega^{0,\bullet}(C,L_D)\,=\, 1-g+\mathrm{deg}(D)\quad.
\end{flalign}

\item[(b)] If $\mathrm{deg}(D)<0$, then $\mathsf{H}^0\, \Omega^{0,\bullet}(C,L_D)\cong 0$.

\item[(c)] If $\mathrm{deg}(D)>2g-2$, then $\mathsf{H}^1\, \Omega^{0,\bullet}(C,L_D)\cong 0$.
\end{itemize} 
\end{theo}

\begin{ex}\label{ex:CP1}
Consider the Riemann sphere $C = \bbC P^1$, which has genus $g=0$.
Combining the results in Theorem \ref{theo:cohomologies}, one obtains 
the following insights into the cohomologies of $\Omega^{0,\bullet}(C,L_D)$:
\begin{itemize}
\item If $\mathrm{deg}(D)< 0$, then $\mathsf{H}^0\, \Omega^{0,\bullet}(C,L_D)\cong 0$
and $\mathsf{H}^1\, \Omega^{0,\bullet}(C,L_D)\cong \bbC^{-(1+\mathrm{deg}(D))}$.

\item If $\mathrm{deg}(D)\geq 0$, then $\mathsf{H}^0\, \Omega^{0,\bullet}(C,L_D)\cong \bbC^{1+\mathrm{deg}(D)}$
and $\mathsf{H}^1\, \Omega^{0,\bullet}(C,L_D)\cong 0$.
\end{itemize}
Such cohomology computations will play a key role
in Section \ref{sec:genuszero} when we explain the conceptual 
mechanisms underlying the relationship between $4d$ semi-holomorphic Chern-Simons theory and 
$2d$ integrable field theories.
\end{ex}

\subsection{\label{subsec:homologicalalgebra}Homological algebra}
$L_\infty$-algebras are a generalization of Lie algebras in the
context of cochain complexes which satisfy a homologically 
relaxed version of the Jacobi identity. The way we use and interpret
$L_\infty$-algebras in our work is as algebraic models for a
certain class of spaces which are called \textit{formal 
moduli problems} \cite{Lurie,Pridham}. In the physical context of field theory, 
these are also known under the name Batalin-Vilkovisky formalism, 
see e.g.\ \cite{CG1,CG2} for a modern perspective and also \cite{JRSW} for a review.
Before we illustrate this crucial link through simple examples,
let us recall the precise definition of an $L_\infty$-algebra
in the sign conventions of \cite{KraftSchnitzer}.
\begin{defi}\label{def:Linftyalgebra}
An \textit{$L_\infty$-algebra} is a pair $(L,\ell)$ consisting
of a $\bbZ$-graded vector space\footnote{In the present paper
all vector spaces will be over the field $\bbC$ of complex numbers.} 
$L = (L^i)_{i\in\bbZ}$ and a family $\ell = (\ell_n)_{n\geq 1} = (\ell_1,\ell_2,\ell_3,\dots)$ 
of graded antisymmetric linear maps $\ell_n : L^{\otimes n}\to L$ of degree
$\vert \ell_n\vert = 2-n$ which satisfy the homotopy Jacobi identities
\begin{flalign}\label{eqn:Linftyalgebra}
\sum_{k+l-1=n}~ \sum_{\sigma\in\mathrm{Sh}(l,k-1)} (-1)^{\vert\sigma\vert} ~(-1)^{k-1}~ 
\ell_k \circ \left(\ell_l \otimes \id^{\otimes (k-1)}\right)\circ \gamma_{\sigma}~=~0\quad,
\end{flalign}
for all $n\geq 1$. Here we denote by $\mathrm{Sh}(l,k-1)\subseteq \Sigma_{n}$
the set of all $(l,k-1)$-shuffle permutations, by $(-1)^{\vert\sigma\vert}$
the parity of the permutation $\sigma\in \mathrm{Sh}(l,k-1)$, and by
$\gamma_\sigma : L^{\otimes n}\to L^{\otimes n}$ its action on tensor powers via
the usual symmetric braiding on graded vector spaces involving Koszul signs.
\end{defi}

\begin{rem}\label{rem:Linftyalgebra}
Note that every $L_\infty$-algebra $(L,\ell)$ has an underlying
cochain complex $(L,\dd)$ with differential $\dd:=\ell_1$ given by the arity $1$ component
of the $L_\infty$-structure $\ell = (\ell_n)_{n\geq 1}$. The square-zero condition
$\dd^2 = 0$ of the differential is equivalent to the homotopy Jacobi identity 
\eqref{eqn:Linftyalgebra} for $n=1$. To simplify our notations below, 
let us denote with a slight abuse of notation the induced differentials on the
tensor powers $L^{\otimes n}$ by the same symbol 
\begin{flalign}
\dd \,:=\, \sum_{i=0}^{n-1} \,\left( \id^{\otimes i} \otimes\dd \otimes\id^{\otimes(n-1-i)}\right) \,:\,
L^{\otimes n}~\longrightarrow~ L^{\otimes n}\quad.
\end{flalign}
The arity $2$ component
of the $L_\infty$-structure is a graded antisymmetric linear map
$\ell_2 : L\otimes L\to L$ of degree $0$, which as a consequence of the homotopy Jacobi identity
\eqref{eqn:Linftyalgebra} for $n=2$ is compatible with the
differential $\dd\,\ell_2 = \ell_2\, \dd$. The linear map $\ell_2$
\textit{does not} in general satisfy the ordinary Jacobi identity of a Lie algebra,
but instead it satisfies the homotopy Jacobi identity \eqref{eqn:Linftyalgebra} for $n=3$.
The latter can be rewritten as
\begin{flalign}
\sum_{\sigma\in \mathrm{Cyc}_3} \ell_2 \circ(\ell_2\otimes\id) \circ \gamma_\sigma\,=\, 
\dd\, \ell_3 + \ell_3\, \dd\quad,
\end{flalign}
where $\mathrm{Cyc}_3\subseteq \Sigma_3$ denotes the set of all cyclic permutations. This means
that the arity $3$ component $\ell_3$ of the $L_\infty$-structure acts as a homotopy witnessing
a homologically relaxed Jacobi identity. The arity $n\geq 4$ components $\ell_n$ provide the necessary higher
homotopy data for this relaxation. Observe that in the case where the homotopy data $\ell_n=0$ are trivial, 
for all $n\geq 3$, the triple $(L,\dd,\ell_2)$ given by the non-vanishing
structure maps defines a differential graded Lie algebra satisfying the Jacobi identity strictly.
Hence, differential graded Lie algebras are special examples of $L_\infty$-algebras.
\end{rem}

\begin{rem}\label{rem:FMP}
The geometric realization of an $L_\infty$-algebra $(L,\ell)$ as a formal moduli problem
$X_{(L,\ell)}$ is obtained through a functor of points which assigns suitable spaces
of Maurer-Cartan elements \cite{Lurie,Pridham}. The resulting geometric picture
can be described intuitively as follows: The points 
of the formal moduli problem $X_{(L,\ell)}$ are given by 
Maurer-Cartan elements in $(L,\ell)$, i.e.\ degree $1$ elements $\alpha \in L^1$
satisfying the Maurer-Cartan equation
\begin{flalign}\label{eqn:MCeqn}
\sum_{n\geq 1} \tfrac{1}{n!}~\ell_n\big(\alpha^{\otimes n}\big) \,=\, 0\quad.
\end{flalign}
(The precise statement requires using thickened points with nilpotents, 
which in particular makes the summation in \eqref{eqn:MCeqn} well-defined.)
These points transform under (higher) gauge symmetries which are encoded in the 
components $L^i$, for all $i\leq 0$. This geometric picture is useful
to interpret the physical content of a field theory which is described
algebraically by an $L_\infty$-algebra.
\end{rem}

\begin{rem}\label{rem:cyclic}
In the context of field theories, the relevant $L_\infty$-algebras $(L,\ell)$ often come endowed
with an additional \textit{cyclic structure} of degree $-3$. Concretely, this consists of a
non-degenerate and graded symmetric linear map $\pair{\cdot}{\cdot} : L\otimes L\to \bbC$ of degree $-3$
which is invariant under the $L_\infty$-structure in the sense that the linear maps
\begin{flalign}
\pair{\cdot}{\cdot}\circ \big(\ell_n\otimes\id \big)\,:\,L^{\otimes (n+1)} ~\longrightarrow~\bbC
\end{flalign}
are graded antisymmetric, for all $n\geq 1$.
From the geometric perspective of Remark \ref{rem:FMP},
such cyclic structures encode $(-1)$-shifted symplectic structures on the associated formal moduli problem $X_{(L,\ell)}$,
while from a physics point of view they provide the additional datum required to define
an action functional and antibracket, see e.g.\ \cite{CG1,CG2,JRSW} for details. 
\end{rem}

Let us illustrate this $L_\infty$-algebraic perspective through some simple examples
of field theories.
\begin{ex}\label{ex:KGfield}
The $L_\infty$-algebra $\big(\F_{\mathrm{KG}}(M),\ell\big)$
describing a scalar field theory with polynomial interactions
on an oriented $d$-dimensional Lorentzian spacetime $M$
is based on the cochain complex
\begin{subequations}
\begin{flalign}
\big(\F_{\mathrm{KG}}(M),\dd\big) \,:=\,\bigg(
\xymatrix@C=7.5em{
\stackrel{(1)}{C^\infty(M)} \ar[r]^-{\dd_M \ast_M \dd_M + m^2\,\ast_M } ~&~\stackrel{(2)}{\Omega^d(M)}
}
\bigg)
\end{flalign}
concentrated in degrees $1$ and $2$, where $\dd_M$ and $\ast_M$
denote, respectively, the de Rham differential and Hodge operator on $M$,
and $m^2$ is a mass term. 
Note that the differential of this complex is related through the Hodge operator
to the usual Klein-Gordon operator $\square_M +m^2$. The arity $n\geq 2$ 
components $\ell_n$ of the $L_\infty$-structure are given by the local interactions
\begin{flalign}
\ell_n\,:\, \F_{\mathrm{KG}}(M)^{\otimes n}~\longrightarrow~\F_{\mathrm{KG}}(M)~~,\quad
\Phi_1\otimes\cdots\otimes\Phi_n ~\longmapsto~\lambda_n\, \ast_M\!\!\big(\Phi_1\cdots\Phi_n\big)
\end{flalign}
\end{subequations}
which are obtained by taking point-wise products of the scalar fields $\Phi_i \in C^\infty(M)$
and applying the Hodge operator,
where $\lambda_n \in \bbC$ denotes coupling constants.
In particular, $\ell_n$ vanishes if any of its arguments belongs to $\Omega^d(M)$.
One easily checks that the homotopy Jacobi identities \eqref{eqn:Linftyalgebra} are satisfied
in this example. To justify the interpretation of this $L_\infty$-algebra in terms of an interacting
scalar field, we use the geometric perspective from Remark \ref{rem:FMP} and observe that
its Maurer-Cartan elements are scalar fields $\Phi\in  \F_{\mathrm{KG}}(M)^1 = C^\infty(M)$
satisfying the interacting Klein-Gordon equation
\begin{flalign}
\dd_M {\ast_M}\dd_M \Phi\,+\, m^2\,{\ast_M}\Phi \,+\, \sum_{n\geq 2} \tfrac{\lambda_n}{n!}\,{\ast_M}\Phi^n\,=\,0\quad.
\end{flalign}
Since the components $\F_{\mathrm{KG}}(M)^i = 0$ are trivial, for all $i\leq 0$, 
this field theory has, as expected, no gauge symmetries.
\sk

The cyclic structure
\begin{subequations}
\begin{flalign}
\pair{\cdot}{\cdot}\,:\,\F_{\mathrm{KG},\cc}(M)\otimes \F_{\mathrm{KG},\cc}(M)~\longrightarrow~\bbC
\end{flalign}
of this model is defined only on compactly supported fields 
(which is a typical feature of field theories on non-compact manifolds)
and it is given by the usual integration pairing
\begin{flalign}
\pair{\Phi}{\Phi^\ddagger}\,:=\,\int_M \Phi\,\Phi^\ddagger
\end{flalign}
\end{subequations}
between fields $\Phi\in \F_{\mathrm{KG},\cc}(M)^1 = C^\infty_\cc(M)$ and antifields 
$\Phi^{\ddagger}\in \F_{\mathrm{KG},\cc}(M)^2 = \Omega^d_\cc(M)$. The resulting action functional
on fields $\Phi\in \F_{\mathrm{KG},\cc}(M)^1$ reads as
\begin{flalign}
\nn S_{\mathrm{KG}}(\Phi)\,&:=\,\Big\langle\!\!\Big\langle \Phi , 
\sum_{n\geq 1} \tfrac{1}{(n+1)!}\,\ell_n\big(\Phi^{\otimes n}\big) 
\Big\rangle\!\!\Big\rangle\\
\,&=\, \int_M\bigg(-\tfrac{1}{2}\,\dd_M\Phi\wedge\ast_M\dd_M\Phi + \tfrac{m^2}{2}\,{\ast_M}\Phi^2  + \sum_{n\geq 2} \tfrac{\lambda_n}{(n+1)!}\, {\ast_M}\Phi^{n+1}\bigg)\quad,
\end{flalign}
which we recognize as the typical action for an interacting scalar field.
\end{ex}

\begin{ex}\label{ex:CSfield}
Let $M$ be an oriented $3$-dimensional manifold and $\g$ a Lie algebra
with Lie bracket denoted by $[\cdot,\cdot]:\g\otimes\g\to \g$.
The $L_\infty$-algebra $\big(\F_{\mathrm{CS}}(M),\ell\big)$ describing the associated Chern-Simons theory 
is based on the $\g$-valued de Rham complex
\begin{subequations}
\begin{flalign}
\big(\F_{\mathrm{CS}}(M),\dd\big) \,:=\,\bigg(
\xymatrix@C=2em{
\stackrel{(0)}{\Omega^0(M,\g)} \ar[r]^-{\dd_M}~&~
\stackrel{(1)}{\Omega^1(M,\g)} \ar[r]^-{\dd_M}~&~
\stackrel{(2)}{\Omega^2(M,\g)} \ar[r]^-{\dd_M}~&~
\stackrel{(3)}{\Omega^3(M,\g)} 
}
\bigg)\quad.
\end{flalign}
The higher arity components of the $L_\infty$-structure
are trivial $\ell_n=0$, for all $n\geq 3$, and the arity 
$2$ component
\begin{flalign}
\ell_2\,:\, \F_{\mathrm{CS}}(M)^{\otimes 2}~\longrightarrow~ \F_{\mathrm{CS}}(M)~~,\quad
\alpha\otimes\beta ~\longmapsto~ \wedgebracket{\alpha}{\beta}
\end{flalign}
\end{subequations}
is given by the usual bracket $\wedgebracket{\alpha}{\beta} := \sum_{ij} \alpha_i\wedge \beta_j\otimes [t_i,t_j]$
which is obtained by combining the $\wedge$-product on differential forms and the Lie bracket on $\g$,
for all $\alpha = \sum_i\alpha_i\otimes t_i$ and $\beta = \sum_j\beta_j\otimes t_j$ in 
$\Omega^\bullet(M,\g) = \Omega^\bullet(M)\otimes \g$.
One easily checks that the homotopy Jacobi identities \eqref{eqn:Linftyalgebra} are satisfied
in this example. To justify the interpretation of this $L_\infty$-algebra in terms of Chern-Simons theory, 
we use the geometric perspective from Remark \ref{rem:FMP} and observe that
its Maurer-Cartan elements are connections $A \in  \F_{\mathrm{CS}}(M)^1 = \Omega^1(M,\g)$
satisfying the flatness condition
\begin{flalign}
\dd_M A + \tfrac{1}{2}\,\wedgebracket{A}{A} \,=\,0\quad.
\end{flalign}
Since the component $\F_{\mathrm{CS}}(M)^0 =\Omega^0(M,\g)$ is non-trivial,
this field theory has, as expected, gauge symmetries which are parametrized by $\g$-valued
functions $\epsilon \in \Omega^0(M,\g) = C^\infty(M,\g)$. Working out explicitly the action 
of these gauge symmetries
on Maurer-Cartan elements, one recovers the usual formula 
$\delta_\epsilon A =\dd_M\epsilon + 
\wedgebracket{A}{\epsilon}$ for infinitesimal gauge transformations of connections,
see e.g.\ \cite{JRSW} for more details.
\sk

To define a cyclic structure $\pair{\cdot}{\cdot}: \F_{\mathrm{CS},\cc}(M)\otimes \F_{\mathrm{CS},\cc}(M)\to \bbC$
for this model, one has to choose any non-degenerate invariant symmetric bilinear form
$\langle\cdot,\cdot\rangle : \g\otimes\g\to \bbC$ on the underlying Lie algebra $\g$.
The cyclic structure is then defined by the integration pairing
\begin{flalign}
\pair{\alpha}{\beta}\,:=\,\int_M \wedgepair{\alpha}{\beta}
\end{flalign}
which combines the $\wedge$-product on differential forms with the pairing $\langle\cdot,\cdot\rangle$ on $\g$.
The resulting action functional on fields $A\in \F_{\mathrm{CS},\cc}(M)^1$ reads as
\begin{flalign}
S_{\mathrm{CS}}(A)\,&:=\,\Big\langle\!\!\Big\langle A , 
\sum_{n\geq 1} \tfrac{1}{(n+1)!}\,\ell_n\big(A^{\otimes n}\big) 
\Big\rangle\!\!\Big\rangle
\,=\, \int_M\Big(
\tfrac{1}{2}\wedgepair{A}{\dd_M A} + \tfrac{1}{3!}\wedgepair{A}{\wedgebracket{A}{A}}
\Big)\quad,
\end{flalign}
which we recognize as the typical Chern-Simons action.
\end{ex}

An important feature of $L_\infty$-algebra structures is that they can be transferred
along deformation retracts of cochain complexes. Let us recall that a 
\textit{strong deformation retract}
\begin{subequations}\label{eqn:defret}
\begin{equation}
\begin{tikzcd}
(L^\prime,\dd^\prime) \ar[r,shift right=-1ex,"i"] & \ar[l,shift right=-1ex,"p"] (L,\dd) \ar[loop,out=-20,in=20,distance=20,swap,"h"]
\end{tikzcd}
\end{equation}
from a cochain complex $(L,\dd)$ to another cochain complex 
$(L^\prime,\dd^\prime)$ consists of two cochain maps $i$ and $p$, i.e.\
$\dd \, i = i\, \dd^\prime$ and $\dd^\prime \, p = p\, \dd$,
and a linear map $h: L\to L$ of degree $-1$, which satisfy the identities
\begin{flalign}
p \, i\,=\,\id~~,\quad
i\, p - \id \,=\,\dd\, h + h\, \dd\quad,
\end{flalign}
as well as the side conditions
\begin{flalign}
h\,i\,=\,0~~,\quad p\,h\,=\,0~~,\quad h^2\,=\,0\quad.
\end{flalign}
\end{subequations}
Note that this implies in particular that $i$ and $p$ are quasi-inverse
to each other, hence both $i$ and $p$ are quasi-isomorphisms of cochain complexes.
The following result is called the \textit{homotopy transfer theorem} for $L_\infty$-algebras,
see e.g.\ \cite[Chapter 10.3]{LodayVallette} for a proof in the  general context of operad algebras.
\begin{theo}\label{theo:homotopytransfer}
Let $(L,\ell)$ be any $L_\infty$-algebra and $(L^\prime,\dd^\prime) \rightleftarrows (L,\dd)\,
\mathbin{\rotatebox[origin=c]{90}{\scalebox{1.3}{$\circlearrowleft$}}}$ any strong deformation retract as in 
\eqref{eqn:defret} from its underlying cochain complex $(L,\dd:=\ell_1)$ 
to another cochain complex $(L^\prime,\dd^\prime)$. Then there exists
a transferred $L_\infty$-algebra structure $\ell^\prime$ on $L^\prime$
with $\ell^\prime_1=\dd^\prime$, such that the cochain map
$i$ extends to an $\infty$-quasi-isomorphism $i_\infty : (L^\prime,\ell^\prime)\rightsquigarrow (L,\ell)$.
In particular, this means that $(L,\ell)$ and $(L^\prime,\ell^\prime)$ represent equivalent $L_\infty$-algebras.
\end{theo}

\begin{rem}\label{rem:homotopytransfer}
The transferred $L_\infty$-algebra structure $\ell^\prime$ can be computed
from the original $L_\infty$-algebra structure $\ell$ and the data of
the strong deformation retract $(i,p,h)$, see e.g.\ \cite[Chapter 10.3]{LodayVallette}. 
In the context of our work, it suffices to consider the 
case where $(L,\ell)$ is a differential graded Lie algebra, i.e.\
the higher arity structure maps $\ell_n=0$ vanish, for all $n\geq 3$.
The transferred arity $2$ component explicitly reads as
\begin{flalign}\label{eqn:ell2transfer}
\ell_2^\prime\,=\,p\circ \ell_2\circ (i\otimes i) \,:\,L^\prime\otimes L^\prime  ~\longrightarrow~
L^\prime\quad.
\end{flalign}
The transferred arity $n\geq 3$ components $\ell^\prime_n$ 
are given by decorating binary rooted trees with the 
given data according to the pattern
\begin{subequations}\label{eqn:ellntransfer}
\begin{flalign}
\begin{gathered}
\begin{tikzpicture}[cir/.style={circle,draw=black,inner sep=0pt,minimum size=2mm},
        poin/.style={rectangle, inner sep=2pt,minimum size=0mm},scale=0.8, every node/.style={scale=0.8},
		every edge quotes/.style = {auto, font=\footnotesize}]
\node[poin] (i)   at (0,-4) {};
\node[poin] (v1)  at (0,-3) {$\ell_2$};
\node[poin] (v2)  at (-1.5,-2) {$\ell_2$};
\node[poin] (v3)  at (1.5,-2) {$\ell_2$};
\node[poin] (v4)  at (3,-1) {$\ell_2$};
\node[poin] (o1)  at (-2.25,-1) {};
\node[poin] (o3)  at (-0.75,-1) {};
\node[poin] (o4)  at (0.5,-1) {};
\node[poin] (o5)  at (2.25,0) {};
\node[poin] (o7)  at (3.75,0) {};
\draw[thick] (i)  edge["$p$"] (v1);
\draw[thick] (v1) edge["$h$"'] (v2);
\draw[thick] (v1) edge["$h$"] (v3);
\draw[thick] (v3) edge["$h$"] (v4);
\draw[thick] (v2) edge["$i$"] (o1);
\draw[thick] (v2) edge["$i$"] (o3);
\draw[thick] (v3) edge["$i$"] (o4);
\draw[thick] (v4) edge["$i$"] (o5);
\draw[thick] (v4) edge["$i$"] (o7);
\end{tikzpicture}
\end{gathered}
\end{flalign}
and then carrying out suitable summations over such trees. Note that the displayed 
decorated tree contributes to $\ell_5^\prime$, since it has five inputs, 
and the associated linear map reads explicitly as
\begin{flalign}
p\circ \ell_2\circ \big((h\circ \ell_2)\otimes (h\circ \ell_2) \big)\circ \big(\id^{\otimes 3}\otimes (h\circ \ell_2)\big) \circ i^{\otimes 5}\quad.
\end{flalign}
\end{subequations}

This schematic description of the transferred $L_\infty$-algebra structure suffices to observe the 
following simplification which arises in some (but not all) of the models studied in this work.
Suppose that the binary structure map $\ell_2 : L\otimes L\to L$
restricts along the (necessarily injective) map $i:L^\prime\to L$, i.e.\
\begin{flalign}
\begin{gathered}
\xymatrix@C=3em{
\ar[d]_-{i\otimes i}L^\prime\otimes L^\prime \ar@{-->}[r]^-{\ell_2}~&~ L^\prime\ar[d]^-{i}\\
L\otimes L\ar[r]_-{\ell_2}~&~L
}
\end{gathered}\quad.
\end{flalign}
Then the transferred arity $2$ component \eqref{eqn:ell2transfer} reads as
\begin{flalign}
\ell_2^\prime \,=\, p\circ \ell_2\circ (i\otimes i) \,= \,p\circ i\circ \ell_2 \,=\, \ell_2\quad,
\end{flalign} 
i.e.\ it agrees with the restriction of $\ell_2$, where we have used that $p\,i = \id$ for any deformation retract. 
From the side condition $h\,i =0$ for a strong deformation retract, 
it follows that the higher arity components $\ell_n^\prime = 0$ are trivial, for all $n\geq 3$.
Hence, in this special scenario the transferred $L_\infty$-algebra
is simply given by the differential graded Lie algebra $(L^\prime,\dd^\prime,\ell_2^\prime=\ell_2)$.
\end{rem}

A useful feature of strong deformation retracts is that they admit deformations by small perturbations.
Recall that a perturbation of the cochain complex $(L,\dd)$ is a linear map $\delta : L\to L$ 
of degree $1$ such that $\dd + \delta$ defines a differential on $L$, i.e.\ it 
satisfies the square-zero condition $(\dd+\delta)^2=0$.
The perturbation $\delta$ is called small, with respect to the given strong deformation retract \eqref{eqn:defret},
if the linear map $\id - \delta\, h : L\to L$ admits an inverse $(\id - \delta\, h)^{-1} : L\to L$.
The following result is called the \textit{homological perturbation lemma}, see e.g.\ \cite{HPT} for a proof.
\begin{propo}\label{prop:HPT}
Given any strong deformation retract \eqref{eqn:defret} and any small perturbation $\delta$ of $(L,\dd)$,
there exists a deformed strong deformation retract
\begin{subequations}
\begin{equation}
\begin{tikzcd}
\big(L^\prime,\dd^\prime +\delta^\prime\big) \ar[r,shift right=-1ex,"\widetilde{i}"] & \ar[l,shift right=-1ex,"\widetilde{p}"] \big(L,\dd+\delta\big) \ar[loop,out=-20,in=20,distance=20,swap,"\widetilde{h}"]
\end{tikzcd}
\end{equation}
with
\begin{flalign}
\delta^\prime \,&:=\,p\, (\id-\delta\, h)^{-1}\, \delta \, i\quad,\\
\widetilde{i}\,&:= \, i + h\, (\id-\delta\, h)^{-1}\, \delta \, i\quad,\\
\widetilde{p}\,&:= \, p + p\, (\id-\delta\, h)^{-1}\, \delta \, h\quad,\\
\widetilde{h}\,&:= \, h + h\, (\id-\delta\, h)^{-1}\, \delta \, h\quad.
\end{flalign}
\end{subequations}
\end{propo}


\section{\label{sec:genuszero}The case of genus zero}
In this section we study the $4d$ semi-holomorphic Chern-Simons
theory of Costello and Yamazaki \cite{CY3} in the homological framework
provided by the theory of $L_\infty$-algebras from Subsection \ref{subsec:homologicalalgebra}.
In particular, we explain how the singularities and boundary conditions
of the fields of this theory, both of which are crucial for establishing links
to $2d$ integrable field theory \cite{DLMV,BSV}, can be implemented conceptually and 
efficiently by using the concept of divisors from Subsection \ref{subsec:Riemannsurfaces}.
To simplify our presentation, we consider in the present section only the case
where the compact Riemann surface $C$ has genus zero, i.e.\ it is given by the
Riemann sphere $C=\CP$. The case of higher genus will be investigated later in Section \ref{sec:highergenus}.

\subsection{\label{subsec:setup}Setup}
The $4d$ semi-holomorphic Chern-Simons theory of Costello and Yamazaki \cite{CY3}
is defined on the $4$-dimensional manifold
\begin{flalign}\label{eqn:Xproduct}
X\,:=\, \Sigma\times C 
\end{flalign}
which is given by the product of an oriented $2d$ Lorentzian spacetime $\Sigma$ 
and a compact Riemann surface $C$. In the present section, the latter 
is chosen as the Riemann sphere $C=\CP$.
Making use of the complex structure on $C$, one can introduce 
the bicomplex
\begin{flalign}\label{eqn:dRDol}
\big( \Omega^{\bullet,(0,\bullet)}(X) , \dd_\Sigma,\delbar \big)\,:=\,
\big(\Omega^\bullet(\Sigma),\dd_\Sigma\big) \widehat{\otimes} \big(\Omega^{0,\bullet}(C),\delbar\big)
\end{flalign}
which combines the de Rham complex on $\Sigma$ with the Dolbeault complex on $C$,
where $\widehat{\otimes}$ denotes the completed projective tensor product
of locally convex topological vector spaces. 
(We regard the de Rham and Dolbeault complexes as endowed with the usual Fr{\'e}chet topologies.)
We will explain in Remark \ref{rem:auxiliaryLinfty} below 
that this de Rham-Dolbeault bicomplex captures the relevant 
semi-holomorphic features of $4d$ Chern-Simons theory \cite{CY3}.
\sk

We will now introduce an $L_\infty$-algebra which provides 
a semi-holomorphic generalization of ordinary Chern-Simons theory from Example \ref{ex:CSfield}.
It is important to emphasize that this $L_\infty$-algebra plays only an auxiliary role
in the context of $4d$ semi-holomorphic Chern-Simons theory since it \textit{does not} (yet)
encode appropriate singularities and boundary conditions for the fields, which however are 
crucial to obtain a non-trivial $2d$ integrable field theory, see e.g.\ 
\cite{CY3,DLMV,BSV}.
These refinements will be discussed in depth in the subsections below.
\begin{defi}\label{def:auxiliaryLinfty}
Let $\g$ be a Lie algebra. We denote by $\big(\E(X),\ell\big)$ the $L_\infty$-algebra 
whose underlying cochain complex
\begin{subequations}
\begin{flalign}
\big(\E(X),\dd\big)\,:=\,\Big(\mathrm{Tot}^\oplus\Big(\Omega^{\bullet,(0,\bullet)}(X,\g)\Big),\dd_\Sigma+\delbar \Big)
\end{flalign}
is given by totalizing the $\g$-valued de Rham-Dolbeault bicomplex \eqref{eqn:dRDol} 
and whose arity $n\geq 2$ structure maps are $\ell_n=0$, for all $n\geq 3$, and 
\begin{flalign}
\ell_2\,:\, \E(X)^{\otimes 2}~\longrightarrow~\E(X)~~,\quad
\alpha\otimes\beta ~\longmapsto~ \wedgebracket{\alpha}{\beta}
\end{flalign}
\end{subequations}
obtained as in Example \ref{ex:CSfield} by combining the $\wedge$-product
with the Lie bracket $[\cdot,\cdot]$ on $\g$.
\end{defi}

\begin{rem}\label{rem:auxiliaryLinfty}
Comparing Definition \ref{def:auxiliaryLinfty} with Example \ref{ex:CSfield},
it is evident that the $L_\infty$-algebra $\big(\E(X),\ell\big)$ is a semi-holomorphic
variant on product manifolds $X=\Sigma\times C$
of the usual Chern-Simons theory on an oriented $3$-manifold $M$. Such types
of models are usually called \textit{topological-holomorphic field theories}
in the literature, see e.g.\ \cite{GRW}. Using the geometric perspective 
from Remark \ref{rem:FMP}, we can provide a more quantitative 
interpretation of this theory: The Maurer-Cartan elements of the
$L_\infty$-algebra $\big(\E(X),\ell\big)$ are pairs
$A \oplus \xi \in \E(X)^1 = \Omega^{1,(0,0)}(X,\g) \oplus \Omega^{0,(0,1)}(X,\g)$
satisfying 
\begin{flalign}
\dd_\Sigma A + \tfrac{1}{2}\,\wedgebracket{A}{A} \,=\,0~~,\quad
\delbar A + \dd_\Sigma \xi + \wedgebracket{A}{\xi}\,=\,0\quad.
\end{flalign}
These equations are familiar from the point of view of
$4d$ semi-holomorphic Chern-Simons theory and $2d$ integrable field theory \cite{CY3}:
The first equation is a flatness condition along $\Sigma$
of the connection $A\in\Omega^{1,(0,0)}(X,\g)$ and the second
equation demands that $A$ is holomorphic, up to a correction
which is determined by the additional field $\xi \in \Omega^{0,(0,1)}(X,\g)$.
This additional field is often denoted by $\xi = A_{\bar{z}}\,\dd \bar{z}$ in the literature.
\end{rem}

We conclude this subsection by substantiating our claim
from above that the $L_\infty$-algebra $\big(\E(X),\ell\big)$ 
from Definition \ref{def:auxiliaryLinfty} is insufficient
in order to capture interesting integrable field theory features of $4d$
semi-holomorphic Chern-Simons theory. For this we use
the homological techniques from Subsection \ref{subsec:homologicalalgebra}
in order to provide an equivalent model for the $L_\infty$-algebra
$\big(\E(X),\ell\big)$ in which the $\delbar$-cohomologies are computed.
(In physics terminology, one could say that we `integrate out' the
fields on the Riemann surface $C=\CP$, resulting in a field theory 
defined only on spacetime $\Sigma$.) In order to carry out these computations,
we choose a strong deformation retract
\begin{equation}\label{eqn:defretauxiliary0}
\begin{tikzcd}
(\bbC,0)\,\cong\,\big(\mathsf{H}^\bullet \,\Omega^{0,\bullet}(C),0\big) 
\ar[r,shift right=-1ex,"i"] & \ar[l,shift right=-1ex,"p"] \big(\Omega^{0,\bullet}(C),\delbar \big) \ar[loop,out=-20,in=20,distance=20,swap,"h"]
\end{tikzcd}
\end{equation}
from the Dolbeault complex to its cohomology, which as a consequence of  
Theorem \ref{theo:cohomologies} is concentrated in degree $0$ and one-dimensional,
see also Example \ref{ex:CP1}. We further assume that this strong deformation retract
is continuous with respect to the usual Fr{\'e}chet topologies
and note that such continuous strong deformation retracts can be obtained for instance 
by using Hodge theory on $C$, see Appendix \ref{app:Hodgetheory}. Continuity
allows us to extend this strong deformation retract along the completed projective tensor product
to a family of strong deformation retracts for the bicomplex \eqref{eqn:dRDol}. 
Applying the totalization construction from Appendix \ref{app:totaldefretract},
one obtains a strong deformation retract
\begin{equation}\label{eqn:defretauxiliary}
\begin{tikzcd}
\big( \E(\Sigma),\dd^\prime\big)\,:=\,\big(\Omega^\bullet(\Sigma,\g),\dd_\Sigma\big)
\ar[r,shift right=-1ex,"i"] & \ar[l,shift right=-1ex,"\widetilde{p}"] \big(\E(X),\dd \big) 
\ar[loop,out=-20,in=20,distance=20,swap,"h"]
\end{tikzcd}
\end{equation}
from the underlying cochain complex of the $L_\infty$-algebra from
Definition \ref{def:auxiliaryLinfty} to the $\g$-valued de Rham complex of $\Sigma$.
Here we have used the simplifications from Remark \ref{rem:deformedtotal},
which clearly apply to our present case at hand: The $i$ map
in \eqref{eqn:defretauxiliary0} is simply the assignment 
$\bbC\hookrightarrow \Omega^{0,\bullet}(C)\,,~\lambda\mapsto \lambda$ 
of constant functions on $C$, hence the de Rham differential $\dd_\Sigma$ along $\Sigma$
restricts along the map $i: \Omega^\bullet(\Sigma,\g)\to
\Omega^{\bullet,(0,\bullet)}(X,\g)$.
Applying now the Homotopy Transfer Theorem \ref{theo:homotopytransfer},
we obtain
\begin{propo}\label{prop:transferauxiliary}
The $L_\infty$-algebra $\big(\E(X),\ell\big)$ from Definition \ref{def:auxiliaryLinfty} 
admits an equivalent description in terms of the $L_\infty$-algebra 
$\big(\E(\Sigma),\ell^\prime\big)$ whose underlying cochain complex
$\big(\E(\Sigma),\dd^\prime\big) = \big(\Omega^\bullet(\Sigma,\g),\dd_\Sigma\big)$ 
is the $\g$-valued de Rham complex of $\Sigma$
and whose higher arity structure maps are given by homotopy transfer 
along the strong deformation retract \eqref{eqn:defretauxiliary}. The transferred
structure maps read explicitly as $\ell_n^\prime = 0$, for all $n\geq 3$, and 
\begin{flalign}
\ell^\prime_2\,:\, \E(\Sigma)^{\otimes 2}~\longrightarrow~\E(\Sigma)~~,\quad
\alpha\otimes\beta ~\longmapsto~ \wedgebracket{\alpha}{\beta}
\end{flalign}
obtained by combining the $\wedge$-product with the Lie bracket $[\cdot,\cdot]$ on $\g$.
\end{propo}
\begin{proof}
This is a direct consequence of the simplification explained
at the end of Remark \ref{rem:homotopytransfer}.
\end{proof}

\begin{rem}\label{rem:transferauxiliary}
Let us provide an interpretation of this result.
The $L_\infty$-algebra $\big(\E(X),\ell\big)$ on the $4$-dimensional manifold $X=\Sigma\times C$  
admits via Proposition \ref{prop:transferauxiliary} an equivalent description
in terms of the $L_\infty$-algebra $\big(\E(\Sigma),\ell^\prime\big)$ which is 
defined intrinsically on the $2$-dimensional spacetime $\Sigma$.
Note that the latter is purely topological and hence it is
insufficient to describe non-trivial integrable field theory features. Indeed, its 
Maurer-Cartan elements are given by connections $A\in \E(\Sigma)^1 = \Omega^1(\Sigma,\g)$ on 
$\Sigma$ satisfying the flatness
condition $\dd_\Sigma A +\frac{1}{2}\,\wedgebracket{A}{A}=0 $, modulo the usual gauge transformations
$\delta_\epsilon A = \dd_\Sigma \epsilon + \wedgebracket{A}{\epsilon}$, 
for $\epsilon \in \Omega^0(\Sigma,\g)$. We will show in the subsections below 
that this (undesired) behavior is drastically altered
by implementing suitable singularities and boundary conditions 
into the $L_\infty$-algebra $\big(\E(X),\ell\big)$.
\end{rem}

\subsection{\label{subsec:singular}Singularities}
Our observation in the previous subsection that the $L_\infty$-algebra
$\big(\E(X),\ell\big)$ from Definition \ref{def:auxiliaryLinfty} does not yet
capture interesting integrability features is not at all surprising:
We have so far neglected one of the key ingredients of $4d$
semi-holomorphic Chern-Simons theory, which is given by a choice
of meromorphic $1$-form $\omega\in \M^{(1)}(C)\setminus\{0\}$ encoding the
meromorphic features of integrable field theories. This meromorphic
$1$-form is used in the literature on $4d$ semi-holomorphic Chern-Simons theory
in multiple, but related, ways. On the one hand, it manifestly enters the action 
functional $S(\A)=\int_X \omega \wedge \mathsf{CS}(\A)$ 
proposed by Costello and Yamazaki \cite{CY3}. On the other hand, 
the zeros and poles of $\omega$ determine special points in the
Riemann surface $C$ at which the field theory develops defects, see e.g.\ 
\cite{CY3,DLMV,BSV}. The concrete pattern is as follows: The fields
are allowed to develop certain singularities at the zeros of $\omega$,
while they are constrained by suitable boundary conditions at the poles of $\omega$.
The main guiding principle underlying the specific choices of the appropriate singularities and 
boundary conditions is the desire to render the action functional 
$S(\A)=\int_X \omega \wedge \mathsf{CS}(\A)$ well-defined and gauge invariant,
see e.g.\ \cite{BSV}.
\sk

In this subsection we focus on the singularities of the fields 
at the zeros of $\omega$ and provide an efficient
description in terms of the theory of divisors and their associated
holomorphic line bundles from Subsection \ref{subsec:Riemannsurfaces}.
Let us fix any meromorphic $1$-form $\omega\in \M^{(1)}(C)\setminus \{0\}$ 
and recall from Example \ref{ex:divisor} that there exists an associated
divisor $(\omega) : C\to \bbZ$ which records the locations and orders of its
zeros and poles. This divisor decomposes
\begin{flalign}\label{eqn:divisordecomposition}
(\omega)\,=\,(\omega)_{0} + (\omega)_{\infty}
\end{flalign}
into the non-negative divisor $(\omega)_{0}\geq 0$
encoding the zeros of $\omega$ and the non-positive divisor
$(\omega)_{\infty}\leq 0$ encoding the poles of $\omega$.
It is worthwhile to recall from \cite[Theorem 17.12]{Forster} that the degree of the 
divisor $(\omega)$ is determined by the genus $g$ of $C$ via $\deg(\omega) = 2g-2$,
hence in our current scenario $C=\CP$ we have that $\deg(\omega)=-2$.
\sk

A key feature of the allowed singularities of the fields 
is that one must treat the individual components 
of differential forms along the spacetime $\Sigma$ differently 
in order to obtain singularities which are compatible with the $\wedge$-product.
In order to implement this feature, we use the Hodge operator
$\ast_\Sigma : \Omega^1(\Sigma)\to \Omega^1(\Sigma)$ on spacetime $\Sigma$
in order to decompose the vector space of $1$-forms
\begin{flalign}
\Omega^1(\Sigma)\,=\,\Omega^+(\Sigma)\oplus \Omega^-(\Sigma)
\end{flalign}
into its chiral and anti-chiral components, i.e.\ $\alpha \in \Omega^\pm(\Sigma)$
if and only if $\ast_\Sigma\alpha = \pm\alpha$. When applied to the 
de Rham-Dolbeault bicomplex $\Omega^{\bullet,(0,\bullet)}(X)$ from 
\eqref{eqn:dRDol}, this yields the 
following decomposition
\begin{flalign}\label{eqn:dRDoldecomposed}
\begin{gathered}
\Omega^{\bullet,(0,\bullet)}(X)\,=\,
\left(\parbox{4cm}{\xymatrix@C=2em@R=0em{
~&~\Omega^{+,(0,\bullet)}(X)  \ar[dr]^-{\dd_\Sigma}~&~ \\
\Omega^{0,(0,\bullet)}(X) \ar[ru]^-{\dd_\Sigma^+}\ar[rd]_-{\dd_\Sigma^-}~&~ \oplus ~&~ \Omega^{2,(0,\bullet)}(X)\\
~&~\Omega^{-,(0,\bullet)}(X)  \ar[ur]_-{\dd_\Sigma}~&~
}}\right)
\end{gathered}\quad.
\end{flalign}
Choosing now any decomposition 
\begin{flalign}\label{eqn:omegapmdecomposition}
(\omega)_{0}\,=\,(\omega)_{0}^+ + (\omega)_{0}^-
\end{flalign}
of the non-negative divisor $(\omega)_{0}\geq 0$ into two non-negative 
divisors $(\omega)_{0}^+ \geq 0$ and $(\omega)_{0}^- \geq 0$, 
we use the associated holomorphic line bundles from Construction \ref{constr:linebundles}
in order to twist the individual components of \eqref{eqn:dRDoldecomposed} according to
\begin{flalign}\label{eqn:dRDolsingular}
\begin{gathered}
\Omega_{\mathrm{sgl}}^{\bullet,(0,\bullet)}(X)\,:=\,
\left(\parbox{4cm}{\xymatrix@C=2em@R=0em{
~&~\Omega^{+,(0,\bullet)}\big(X,L_{(\omega)_0^+}\big)  \ar[dr]^-{\dd_\Sigma}~&~ \\
\Omega^{0,(0,\bullet)}(X) \ar[ru]^-{\dd_\Sigma^+}\ar[rd]_-{\dd_\Sigma^-}~&~ \oplus ~&~ \Omega^{2,(0,\bullet)}\big(X,L_{(\omega)_0}\big)\\
~&~\Omega^{-,(0,\bullet)}\big(X,L_{(\omega)_0^-}\big)  \ar[ur]_-{\dd_\Sigma}~&~
}}\right)
\end{gathered}\quad,
\end{flalign}
where in the displayed differentials we suppress the holomorphic line bundle morphisms
from Construction \ref{constr:linebundles} which are associated with
$0\leq (\omega)_0^\pm$ and $ (\omega)_0^\pm \leq (\omega)_0$. Note that the key feature
of the decomposition \eqref{eqn:omegapmdecomposition} and the specific 
assignment of holomorphic line bundles in \eqref{eqn:dRDolsingular} 
is that the $\wedge$-product is well-defined on $\Omega_{\mathrm{sgl}}^{\bullet,(0,\bullet)}(X)$ 
as a consequence of 
\begin{flalign}
L_{(\omega)_0^+}\otimes L_{(\omega)_0^-}\,\cong\,L_{(\omega)_0^+ + (\omega)_0^-}\, \cong\, L_{(\omega)_0}\quad.
\end{flalign}
With these preparations, we can now define an improved variant of the $L_\infty$-algebra
from Definition \ref{def:auxiliaryLinfty} which encodes the desired singularities of the fields.
\begin{defi}\label{def:singularLinfty}
Let $\g$ be a Lie algebra and $\omega\in \M^{(1)}(C)\setminus \{0\}$ a
meromorphic $1$-form. For any choice of decomposition as 
in \eqref{eqn:omegapmdecomposition}, we denote by $\big(\L(X),\ell\big)$ the 
$L_\infty$-algebra whose underlying cochain complex
\begin{subequations}
\begin{flalign}
\big(\L(X),\dd\big)\,:=\,\Big(\mathrm{Tot}^\oplus\Big(\Omega_{\mathrm{sgl}}^{\bullet,(0,\bullet)}\big(X,\g\big)\Big),\dd_\Sigma+\delbar \Big)
\end{flalign}
is given by totalizing the $\g$-valued twisted de Rham-Dolbeault bicomplex \eqref{eqn:dRDolsingular} 
and whose arity $n\geq 2$ structure maps are $\ell_n=0$, for all $n\geq 3$, and 
\begin{flalign}
\ell_2\,:\, \L(X)^{\otimes 2}~\longrightarrow~\L(X)~~,\quad
\alpha\otimes\beta ~\longmapsto~ \wedgebracket{\alpha}{\beta}
\end{flalign}
\end{subequations}
obtained by combining the $\wedge$-product with the Lie bracket $[\cdot,\cdot]$ on $\g$.
\end{defi}

In order to obtain a better understanding of this field theory,
let us use again the homological techniques from Subsection \ref{subsec:homologicalalgebra}
in order to determine an equivalent model for the $L_\infty$-algebra
$\big(\L(X),\ell\big)$ in which the $\delbar$-cohomologies are computed, i.e.\
the fields on the Riemann surface $C=\CP$ are `integrated out'. In analogy 
to the previous subsection, we do this by choosing continuous strong
deformation retracts from the (twisted) Dolbeault complexes to their cohomologies,
which can be obtained for instance via Hodge theory, see Appendix \ref{app:Hodgetheory}.
Since all divisors appearing in \eqref{eqn:dRDolsingular} are of degree $\geq 0$,
the results in Theorem \ref{theo:cohomologies}, see also Example \ref{ex:CP1},
imply that these strong deformation retracts are of the form
\begin{equation}\label{eqn:defretsingular0}
\begin{tikzcd}
\big(\O_D(C),0\big) \ar[r,shift right=-1ex,"i_D"] & \ar[l,shift right=-1ex,"p_D"] \big(\Omega^{0,\bullet}(C,L_D),\delbar \big) \ar[loop,out=-20,in=20,distance=30,swap,"h_D"]
\end{tikzcd}\qquad \big(\text{for $\deg(D)\geq 0$}\big)\quad.
\end{equation}
The map $i_D$ is given in terms of the isomorphism \eqref{eqn:O_DGammaisos} by
$i_D : \O_D(C) \overset{\cong}\to \O(C,L_D) \subset \Omega^{0,\bullet}(C,L_D)$.
Extending these continuous strong deformation retracts
along the completed projective tensor product
to a family of strong deformation retracts for the bicomplex \eqref{eqn:dRDolsingular},
we can apply the totalization construction from Appendix \ref{app:totaldefretract}.
From commutative diagrams as in \eqref{eqn:DtoDprimediagram}, it follows that
the simplification from Remark \ref{rem:deformedtotal} applies to the present case,
leading to a strong deformation retract
\begin{subequations}\label{eqn:defretsingular}
\begin{equation}\label{eqn:defretsingular1}
\begin{tikzcd}
\big(\L(\Sigma),\dd^\prime\big)
\ar[r,shift right=-1ex,"i"] & \ar[l,shift right=-1ex,"\widetilde{p}"] \big(\L(X),\dd \big) 
\ar[loop,out=-20,in=20,distance=20,swap,"h"]
\end{tikzcd}
\end{equation}
from the underlying cochain complex of the $L_\infty$-algebra $\big(\L(X),\ell\big)$
from Definition \ref{def:singularLinfty} to the cochain complex
\begin{flalign}\label{eqn:defretsingular2}
\begin{gathered}
\big(\L(\Sigma),\dd^\prime\big)\,:=\,
\left(\parbox{4cm}{\xymatrix@C=2em@R=0em{
~&~\Omega^{+}(\Sigma,\g) \otimes \O_{(\omega)_0^+}(C) \ar[dr]^-{\dd_\Sigma}~&~ \\
\Omega^{0}(\Sigma,\g) \ar[ru]^-{\dd_\Sigma^+}\ar[rd]_-{\dd_\Sigma^-}~&~ \oplus ~&~ \Omega^{2}(\Sigma,\g) 
\otimes \O_{(\omega)_0}(C)\\
~&~\Omega^{-}(\Sigma,\g) \otimes \O_{(\omega)_0^-}(C) \ar[ur]_-{\dd_\Sigma}~&~
}}\right)\quad.
\end{gathered}
\end{flalign}
\end{subequations}
Comparing this result with the previous case \eqref{eqn:defretauxiliary} 
where no singularities of the fields are allowed, we observe that the new
feature is the appearance of the (finite-dimensional) 
vector spaces $\O_D(C)$ of divisor-conditioned meromorphic functions,
for all three divisors entering Definition \ref{def:singularLinfty}.
Applying now the Homotopy Transfer Theorem \ref{theo:homotopytransfer},
we obtain
\begin{propo}\label{prop:transfersingular}
The $L_\infty$-algebra $\big(\L(X),\ell\big)$ from Definition \ref{def:singularLinfty}
admits an equivalent description in terms of the $L_\infty$-algebra 
$\big(\L(\Sigma),\ell^\prime\big)$ whose underlying cochain complex
is given by \eqref{eqn:defretsingular2}
and whose higher arity structure maps are given by homotopy transfer 
along the strong deformation retract \eqref{eqn:defretsingular1}. The transferred
structure maps read explicitly as $\ell_n^\prime = 0$,
for all $n\geq 3$, and 
\begin{flalign}
\ell^\prime_2\,:\, \L(\Sigma)^{\otimes 2}~\longrightarrow~\L(\Sigma)~~,\quad
\alpha\otimes\beta ~\longmapsto~ \wedgebracket{\alpha}{\beta}
\end{flalign}
obtained by combining the $\wedge$-product with the Lie bracket $[\cdot,\cdot]$ on $\g$.
\end{propo}
\begin{proof}
This is a direct consequence of the simplification explained
at the end of Remark \ref{rem:homotopytransfer}.
\end{proof}

\begin{rem}\label{rem:transfersingular}
The $L_\infty$-algebra $\big(\L(\Sigma),\ell^\prime\big)$  from 
Proposition \ref{prop:transfersingular} describes much more interesting 
features related to integrability than the one without singularities 
from Proposition \ref{prop:transferauxiliary}, see also Remark \ref{rem:transferauxiliary}. 
Indeed, its Maurer-Cartan elements are given by the chiral components
\begin{subequations}
\begin{flalign}
A_+\oplus A_- \,\in\, \L(\Sigma)^1\,=\, 
\Big(\Omega^{+}(\Sigma,\g) \otimes \O_{(\omega)_0^+}(C)\Big) 
\oplus \Big(\Omega^{-}(\Sigma,\g) \otimes \O_{(\omega)_0^-}(C)\Big)
\end{flalign}
of a connection $1$-form $A = A_+ + A_-\in \Omega^{1}(\Sigma,\g)\widehat{\otimes}\M(C)$
on $\Sigma$ which is allowed to depend meromorphically on $C$.
Note that the poles of $A$ are controlled by the divisors $(\omega)_0^+$ and $(\omega)_0^-$. 
The Maurer-Cartan equation is the flatness condition
\begin{flalign}
\dd_\Sigma A +\tfrac{1}{2}\,\wedgebracket{A}{A}\,=\,0
\end{flalign}
\end{subequations}
along $\Sigma$. Observing that these are precisely the desired features
of a Lax connection for a $2d$ integrable field theory, one can interpret
the $L_\infty$-algebra $\big(\L(\Sigma),\ell^\prime\big)$ as describing a theory of Lax connections.
It is worthwhile to highlight that the gauge transformations 
$\delta_\epsilon A = \dd_{\Sigma}\epsilon + \wedgebracket{A}{\epsilon}$ are 
given by parameters $\epsilon\in\Omega^0(\Sigma,\g)$ that are constant 
along the Riemann surface $C$, which limits the amount of gauge freedom
in the theory of Lax connections $A$. Loosely speaking, 
one can say that only the zero-mode $A_0$ of $A$ (i.e.\ its constant part along $C$) transforms 
as a gauge field $\delta_\epsilon A_0 = \dd_{\Sigma}\epsilon + \wedgebracket{A_0}{\epsilon}$,
while the higher modes $A_{>0}$ of $A$ (i.e.\ its non-constant parts along $C$) 
transform as matter fields in the adjoint representation $\delta_{\epsilon} A_{>0} = \wedgebracket{A_{>0}}{\epsilon}$.
\end{rem}

\subsection{\label{subsec:boundary}Boundary conditions}
While the singularities of fields from the previous subsection
are located at the zeros of the given meromorphic 
$1$-form $\omega\in\M^{(1)}(C)\setminus\{0\}$,
boundary conditions have to be imposed at the poles of $\omega$.
In the language of divisors, this means that for the study of boundary conditions
the non-positive part $(\omega)_{\infty}\leq 0$ of the divisor 
\eqref{eqn:divisordecomposition} is relevant. We will now discuss a class 
of \textit{local} boundary conditions, which includes the boundary conditions
considered in \cite{CY3} but is less general than the (in general non-local on $C$) 
isotropic boundary conditions from \cite{BSV}. These local boundary conditions
have the advantage that they can be implemented and studied very efficiently 
by using divisors and their associated holomorphic line bundles.
\begin{defi}\label{def:bdycondition}
Fix any meromorphic $1$-form $\omega\in \M^{(1)}(C)\setminus \{0\}$ 
and any choice of decomposition $(\omega)_0 = (\omega)_0^+ + (\omega)_0^-$ as in \eqref{eqn:omegapmdecomposition}.
A \textit{local boundary condition} is a choice of partially ordered non-positive divisors
\begin{subequations}\label{eqn:bdycondition}
\begin{equation}\label{eqn:bdycondition1}
\begin{tikzcd}[column sep=-0.7em, row sep=-0.7em]
&& (\omega)_\infty^+  && \\
&\mathbin{\rotatebox[origin=c]{45}{$\leq$}} & &\mathbin{\rotatebox[origin=c]{-45}{$\leq$}} &\\
(\omega)_\infty^0  &&  &&(\omega)_\infty^2 \,\leq \,0\\
&\mathbin{\rotatebox[origin=c]{-45}{$\leq$}} & &\mathbin{\rotatebox[origin=c]{45}{$\leq$}} &\\
&& (\omega)_\infty^- &&
\end{tikzcd}\quad,
\end{equation}
such that
\begin{flalign}\label{eqn:bdycondition2}
(\omega)_\infty^0 + (\omega)_\infty^2 \,=\,(\omega)_\infty\,=\,(\omega)_\infty^+ + (\omega)_\infty^-
\end{flalign}
and
\begin{flalign}\label{eqn:bdycondition3}
\deg\big((\omega)_0^+ +(\omega)_\infty^+\big) \,=\,\deg\big((\omega)_0^- +(\omega)_\infty^-\big)\,=\,g-1\quad,
\end{flalign}
\end{subequations} 
where $g$ denotes the genus of $C$. (In the present case $C=\CP$, we have that $g=0$.)
\end{defi}

\begin{rem}
Let us provide a motivation for the boundary conditions in Definition \ref{def:bdycondition}.
The partial ordering of divisors in \eqref{eqn:bdycondition1} is essential to define
the twisted de Rham-Dolbeault bicomplex \eqref{eqn:dRDolboundary} which encodes
both singularities and boundary conditions of the fields. The condition \eqref{eqn:bdycondition2}
is a kind of `saturation condition' which is needed to obtain a non-degenerate cyclic structure
and action functional, see Subsection \ref{subsec:cyclic} below. The last condition 
\eqref{eqn:bdycondition3} can be motivated by counting degrees of freedom:
The chiral components of the Lax connection from Proposition \ref{prop:transfersingular}
and Remark \ref{rem:transfersingular} take values in $\O_{(\omega)_0^\pm}(C)$.
Using Theorem \ref{theo:cohomologies} and Example \ref{ex:CP1},
one finds that 
\begin{flalign}
\O_{(\omega)_0^\pm}(C) \,\cong\, \bbC^{1+\deg(\omega)_0^\pm}\quad,
\end{flalign}
hence one should impose $1+\deg(\omega)_0^\pm$ many boundary conditions in order 
to fix the Lax connection by the fields of an integrable field theory.
This is precisely what one gets by rearranging 
the conditions \eqref{eqn:bdycondition3} as
\begin{flalign}
\deg(\omega)_\infty^{\pm} \,=\, -1 - \deg(\omega)_0^{\pm}
\end{flalign} 
and observing that the number of boundary conditions is minus the 
degree of $(\omega)_\infty^{\pm}\leq 0$.
\end{rem}

\begin{ex}
The prime example of a meromorphic $1$-form on $C=\CP$ 
which is relevant for integrable field theory is given by
\begin{flalign}
\omega\,=\,\frac{(1-z^2)}{z^2}~\dd z\quad,
\end{flalign}
where $z$ is the standard holomorphic coordinate on $\CP\setminus\{\infty\}$.
This meromorphic $1$-form has two simple zeros at $z=1$ and $z=-1$, 
and two double poles at $z=0$ and $z=\infty$. 
Splitting the non-negative divisor $(\omega)_0 = \delta_{1} + \delta_{-1}$ 
according to $(\omega)_{0}^\pm = \delta_{\pm 1}$, a possible choice
of local boundary condition in the sense of Definition \ref{def:bdycondition} is given by
\begin{flalign}
(\omega)_{\infty}^0\,=\,(\omega)_{\infty}^\pm \,=\, (\omega)_{\infty}^2 
\,=\, \tfrac{1}{2}\, (\omega)_\infty \,=\, -\delta_{0}-\delta_{\infty}\quad.
\end{flalign}
This choice is relevant for the construction of the principal chiral model, 
see e.g.\ \cite[Section 7]{Lacroix} for a well-presented review.
\end{ex}

Given any local boundary condition as in Definition \ref{def:bdycondition},
we use the associated holomorphic line bundles from Construction \ref{constr:linebundles} 
to twist the components of the de Rham-Dolbeault bicomplex with singularities from \eqref{eqn:dRDolsingular},
which leads to the bicomplex
\begin{flalign}\label{eqn:dRDolboundary}
\begin{gathered}
\Omega_{\mathrm{sgl},\mathrm{bdy}}^{\bullet,(0,\bullet)}(X)\,:=\,
\left(\parbox{4cm}{\xymatrix@C=1em@R=0em{
~&~\Omega^{+,(0,\bullet)}\big(X,L_{(\omega)_0^+ + (\omega)_{\infty}^+}\big)  \ar[dr]^-{\dd_\Sigma}~&~ \\
\Omega^{0,(0,\bullet)}\big(X,L_{(\omega)_{\infty}^0}\big) \ar[ru]^-{\dd_\Sigma^+}\ar[rd]_-{\dd_\Sigma^-}~&~ \oplus ~&~ \Omega^{2,(0,\bullet)}\big(X,L_{(\omega)_0 + (\omega)_{\infty}^2}\big)\\
~&~\Omega^{-,(0,\bullet)}\big(X,L_{(\omega)_0^- + (\omega)_{\infty}^-}\big)  \ar[ur]_-{\dd_\Sigma}~&~
}}\right)
\end{gathered}~~.
\end{flalign}
Making use of the holomorphic line bundle morphisms associated with 
partially ordered divisors from Construction \ref{constr:linebundles},
one easily checks that the $\wedge$-product is well-defined on 
$\Omega_{\mathrm{sgl},\mathrm{bdy}}^{\bullet,(0,\bullet)}(X)$.
With these preparations, we can now define an improved variant of the $L_\infty$-algebras
from Definitions \ref{def:auxiliaryLinfty} and \ref{def:singularLinfty}
which encodes both the desired singularities and boundary conditions of the fields.
\begin{defi}\label{def:boundaryLinfty}
Let $\g$ be a Lie algebra and $\omega\in \M^{(1)}(C)\setminus \{0\}$ a
meromorphic $1$-form. For any choice of decomposition as 
in \eqref{eqn:omegapmdecomposition} and local boundary condition as in Definition \ref{def:bdycondition}, 
we denote by $\big(\F(X),\ell\big)$ the $L_\infty$-algebra whose underlying cochain complex
\begin{subequations}
\begin{flalign}
\big(\F(X),\dd\big)\,:=\,\Big(\mathrm{Tot}^\oplus\Big(\Omega_{\mathrm{sgl},\mathrm{bdy}}^{\bullet,(0,\bullet)}\big(X,\g\big)\Big),\dd_\Sigma+\delbar \Big)
\end{flalign}
is given by totalizing the $\g$-valued twisted de Rham-Dolbeault bicomplex \eqref{eqn:dRDolboundary} 
and whose arity $n\geq 2$ structure maps are $\ell_n=0$, for all $n\geq 3$, and 
\begin{flalign}
\ell_2\,:\, \F(X)^{\otimes 2}~\longrightarrow~\F(X)~~,\quad
\alpha\otimes\beta ~\longmapsto~ \wedgebracket{\alpha}{\beta}
\end{flalign}
\end{subequations}
obtained by combining the $\wedge$-product with the Lie bracket $[\cdot,\cdot]$ on $\g$.
\end{defi}

Let us now analyze this field theory from the perspective of spacetime $\Sigma$
by computing the $\delbar$-cohomologies. The new feature of the bicomplex 
\eqref{eqn:dRDolboundary}, in comparison with the cases discussed in the previous subsections,
is the presence of non-positive divisors $(\omega)_{\infty}^0\leq (\omega)_{\infty}^\pm \leq (\omega)_{\infty}^2\leq 0$. 
In order to apply the results about twisted Dolbeault cohomologies from Theorem \ref{theo:cohomologies} and
Example \ref{ex:CP1}, let us start with determining the degrees of the individual divisors appearing
in \eqref{eqn:dRDolboundary}:
\begin{itemize}
\item The divisor entering the $0$-form component along $\Sigma$ of the bicomplex 
\eqref{eqn:dRDolboundary} has degree
\begin{flalign}\label{eqn:degestimate}
\deg(\omega)_{\infty}^0 \,\leq\, -1\quad.
\end{flalign}
This follows from the conditions $(\omega)_{\infty}^0\leq (\omega)_{\infty}^2\leq 0$ and
$(\omega)_{\infty} = (\omega)_{\infty}^0 + (\omega)_{\infty}^2$ from Definition \ref{def:bdycondition}, 
together with the fact that $\deg(\omega)_\infty\leq -2$ since any meromorphic $1$-form on 
the Riemann sphere $C=\CP$ has at least two poles, see e.g.\ \cite[Theorem 17.12]{Forster}.
This implies via Example \ref{ex:CP1} that the corresponding twisted Dolbeault cohomology 
vanishes in degree $0$ and that it has dimension $-1-\deg(\omega)_{\infty}^0$ in degree $1$. 

\item By Definition \ref{def:bdycondition}, the divisors entering the $\pm$-form components along
$\Sigma$ of the bicomplex \eqref{eqn:dRDolboundary} have degrees 
\begin{flalign}
\deg\left((\omega)_0^\pm + (\omega)_{\infty}^\pm \right)\,=\,g-1 \,=\,-1\quad.
\end{flalign}
This implies via Example \ref{ex:CP1}
that the corresponding twisted Dolbeault cohomologies vanish in all degrees.

\item The divisor entering the $2$-form component along $\Sigma$ of the bicomplex 
\eqref{eqn:dRDolboundary} has degree
\begin{flalign}
\deg\left((\omega)_0 + (\omega)_{\infty}^2\right) \,=\, 
\deg\left((\omega)_0 + (\omega)_{\infty} - (\omega)_{\infty}^0\right)\,=\,
-2 - \deg(\omega)_\infty^0 \,\geq -1\quad,
\end{flalign}
where the first step uses \eqref{eqn:bdycondition2}, the second step uses
$\deg(\omega)=2g -2=-2$ (see e.g.\ \cite[Theorem 17.12]{Forster}) and
the last step uses the estimate \eqref{eqn:degestimate}.
This implies via Example \ref{ex:CP1} that the corresponding twisted Dolbeault cohomology 
vanishes in degree $1$ and that it has dimension $-1- \deg(\omega)_\infty^0$ in degree $0$. 
\end{itemize}
\begin{notation}\label{not:N}
To simplify the notations below, let us introduce the non-negative integer
\begin{flalign}
N\,:=\,-1- \deg(\omega)_\infty^0 \,\geq\,0
\end{flalign}
to denote the dimension of the not necessarily vanishing twisted Dolbeault cohomologies.
Note that in the case where $\omega$ has at least three poles (counting multiplicities), 
this integer is positive $N>0$.
\end{notation}

In analogy to the previous subsections, we choose now continuous strong deformation retracts
from the twisted Dolbeault complexes to their cohomologies, which can be obtained for instance
via Hodge theory, see Appendix \ref{app:Hodgetheory}. Using our insights obtained from the 
degree counting exercise above, these continuous strong deformation retracts take the form
\begin{subequations}\label{eqn:defretboundary0}
\begin{equation}\label{eqn:defretboundary00}
\begin{tikzcd}
\big(\bbC^{N}[-1],0\big) \ar[r,shift right=-1ex,"i^0"] & \ar[l,shift right=-1ex,"p^0"] \Big(\Omega^{0,\bullet}\big(C,L_{(\omega)_{\infty}^0}\big),\delbar \Big) \ar[loop,out=-20,in=20,distance=35,swap,"h^0"]
\end{tikzcd} 
\end{equation}
for the $0$-form component along $\Sigma$, 
where $[-1]$ is the standard notation
for a degree shift by $+1$,
\begin{equation}\label{eqn:defretboundary0pm}
\begin{tikzcd}
\big(0,0\big) \ar[r,shift right=-1ex,"i^\pm"] & \ar[l,shift right=-1ex,"p^\pm"] \Big(\Omega^{0,\bullet}\big(C,L_{(\omega)_0^\pm + (\omega)_{\infty}^\pm}\big),\delbar \Big) \ar[loop,out=-20,in=20,distance=50,swap,"h^\pm"]
\end{tikzcd} 
\end{equation}
for the $\pm$-form components along $\Sigma$, and
\begin{equation}\label{eqn:defretboundary02}
\begin{tikzcd}
\big(\bbC^{N},0\big) \ar[r,shift right=-1ex,"i^2"] & \ar[l,shift right=-1ex,"p^2"] \Big(\Omega^{0,\bullet}\big(C,L_{(\omega)_0 + (\omega)_{\infty}^2}\big),\delbar \Big) \ar[loop,out=-20,in=20,distance=50,swap,"h^2"]
\end{tikzcd} 
\end{equation}
\end{subequations}
for the $2$-form component along $\Sigma$. 
Extending these continuous strong  deformation retracts
along the completed projective tensor product
to a family of strong deformation retracts for the bicomplex \eqref{eqn:dRDolboundary},
we can apply the totalization construction from Appendix \ref{app:totaldefretract}.
In contrast to the previous subsections, in the present case
the simplification from Remark \ref{rem:deformedtotal} \textit{does not} apply
because the cohomology in \eqref{eqn:defretboundary00} is non-trivial in degree $1$, while 
the cohomologies \eqref{eqn:defretboundary0pm} are trivial. This implies that we have to work
with the non-simplified variant \eqref{eqn:deformedtotal} of the totalization construction.
Working out the details leads to a strong deformation retract
\begin{subequations}\label{eqn:defretboundary}
\begin{equation}\label{eqn:defretboundary1}
\begin{tikzcd}
\big(\F(\Sigma),\dd^\prime\big)
\ar[r,shift right=-1ex,"\widetilde{i}"] & \ar[l,shift right=-1ex,"\widetilde{p}"] \big(\F(X),\dd \big) 
\ar[loop,out=-20,in=20,distance=20,swap,"h"]
\end{tikzcd}
\end{equation}
from the underlying cochain complex of the $L_\infty$-algebra $\big(\F(X),\ell\big)$
from Definition \ref{def:boundaryLinfty} to the cochain complex
\begin{flalign}\label{eqn:defretboundary2}
\big(\F(\Sigma),\dd^\prime\big)\,:=\,
\bigg(
\xymatrix@C=9em{
\stackrel{(1)}{\Omega^0\big(\Sigma,\g^{N}\big) } \ar[r]^-{p^2\,\dd_\Sigma \,\big(h^+ \dd_{\Sigma}^+ 
+ h^- \dd_{\Sigma}^- \big) \,i^0} ~&~\stackrel{(2)}{\Omega^2\big(\Sigma,\g^{N}\big)}
}
\bigg)\quad.
\end{flalign}
\end{subequations}
Comparing this result with the previous cases in
\eqref{eqn:defretauxiliary} and \eqref{eqn:defretsingular}, we observe that the key new features
are that the complex \eqref{eqn:defretboundary2} is concentrated only in degrees $1$ and $2$,
and that its differential is a second-order differential operator on spacetime $\Sigma$.
Applying now the Homotopy Transfer Theorem \ref{theo:homotopytransfer},
we can summarize our construction in this subsection as follows.
\begin{propo}\label{prop:transferboundary}
The $L_\infty$-algebra $\big(\F(X),\ell\big)$ from Definition \ref{def:boundaryLinfty}
admits an equivalent description in terms of the $L_\infty$-algebra 
$\big(\F(\Sigma),\ell^\prime\big)$ whose underlying cochain complex
is given by \eqref{eqn:defretboundary2}
and whose higher arity structure maps are given by homotopy transfer 
along the strong  deformation retract \eqref{eqn:defretboundary1}. 
\end{propo}

\begin{rem}\label{rem:transferboundary}
The $L_\infty$-algebra $\big(\F(\Sigma),\ell^\prime\big)$ 
from Proposition \ref{prop:transferboundary} captures precisely 
the key features of a family of $N=-1-\deg(\omega)_\infty^0\geq 0$ 
many $\g$-valued scalar fields on spacetime $\Sigma$ which interact 
through $\sigma$-model-type local interactions. To justify this statement, 
consider the Maurer-Cartan equation
\begin{flalign}
\dd^\prime \Phi  + \sum_{n\geq 2}\tfrac{1}{n!} \,\ell^\prime_n\big(\Phi^{\otimes n}\big) \,=\,0
\end{flalign}
for the fields $\Phi\in \F(\Sigma)^1 = \Omega^0(\Sigma,\g^{N}) $.
The linear term $\dd^\prime \Phi$ is given by a second-order linear differential operator 
\eqref{eqn:defretboundary2} on $\Sigma$, which depends on the Lorentzian metric
through the different treatment of the $\pm$-chiral components of $1$-forms.
The structure of this differential operator resembles the one of the d'Alembertian (see Example \ref{ex:KGfield})
in $2d$, which can be written as $\dd_\Sigma\ast_\Sigma\dd_\Sigma
= \dd_\Sigma\ast_\Sigma(\dd_\Sigma^+ + \dd_\Sigma^-) = \dd_\Sigma \,(\dd_\Sigma^+ - \dd_\Sigma^-)$,
however its detailed description is governed by the deformation retracts.
The interaction terms $\ell^\prime_n\big(\Phi^{\otimes n}\big)$ are determined by the Homotopy 
Transfer Theorem \ref{theo:homotopytransfer} and their detailed description
will depend on the choices of $\omega$, the decomposition \eqref{eqn:omegapmdecomposition}
and the local boundary condition from Definition \ref{def:bdycondition}.
(More informally, one can say that these data
dictate the `coupling constants' for these interactions.)
It is important to observe that the strong deformation retract \eqref{eqn:defretboundary1} which determines these
transferred structure maps is local on $\Sigma$: Indeed, the component-wise strong 
deformation retracts \eqref{eqn:defretboundary0} are independent of $\Sigma$
and the totalization construction \eqref{eqn:deformedtotal} introduces 
at most a single differential $\dd_\Sigma$ in every map of 
the deformed strong deformation retract $(\widetilde{i},\widetilde{p},h)$.
From our schematic description of the transferred structure maps in Remark \ref{rem:homotopytransfer},
one then sees immediately that the terms $\ell^\prime_n\big(\Phi^{\otimes n}\big)$ 
are determined locally on $\Sigma$ by $\Phi$ and its first and second derivatives along $\Sigma$.
Counting differential form degrees, one additionally observes that $\ell^\prime_n\big(\Phi^{\otimes n}\big)
\in \Omega^2(\Sigma,\g^{N})$ 
involves precisely a total number of two derivatives along $\Sigma$, distributed
among the $n$ factors $\Phi^{\otimes n}$. Furthermore, since $\F(\Sigma)^i=0$, for all $i\leq 0$, 
it follows that this field theory does not have any non-trivial gauge symmetries.
So we obtain qualitatively all the key
features of a (perturbative) $\sigma$-model on spacetime $\Sigma$.
\end{rem}

\subsection{\label{subsec:Lax}Lax connections via $L_\infty$-morphisms}
It is natural to expect that the $L_\infty$-algebra $\big(\F(\Sigma),\ell^\prime\big)$ 
from Proposition \ref{prop:transferboundary} will describe a $2d$ integrable field theory 
on spacetime $\Sigma$ since it has been obtained from $4d$ semi-holomorphic Chern-Simons theory
on $X=\Sigma\times C$. We have already recognized in Remark 
\ref{rem:transferboundary} that this theory displays all
the relevant key feature of a $\sigma$-model-type field theory on spacetime $\Sigma$,
but it remains to clarify if, and in which sense, this theory is integrable. Ideally, one would
like to construct a suitable Lax connection for this field theory.
\sk

Our homological framework provides an elegant and conceptual way to answer these questions
by exhibiting manifestly a Lax connection. The key idea can be summarized by the diagram
\begin{flalign}\label{eqn:maindiagram}
\begin{gathered}
\xymatrix@C=4em{
 \ar@{~>}[d]_-{i_\infty}^-{\sim} \big(\F(\Sigma),\ell^\prime\big) \ar@{~>}[r]^-{}~&~\big(\L(\Sigma),\ell^\prime\big)  \\
\big(\F(X),\ell\big)~ \ar[r]&~ \big(\L(X),\ell\big) \ar@{~>}[u]_-{i^{-1}_\infty}^-{\sim}
}
\end{gathered}
\end{flalign}
which we shall now explain in detail: 
\begin{itemize}
\item On the left-hand side, we have the $L_\infty$-algebra
$\big(\F(X),\ell\big)$ from Definition \ref{def:boundaryLinfty}, which describes 
$4d$ semi-holomorphic Chern-Simons theory with both singularities and boundary conditions,
and its equivalent model $\big(\F(\Sigma),\ell^\prime\big)$ from 
Proposition \ref{prop:transferboundary}, which as explained in Remark 
\ref{rem:transferboundary} describes a $\sigma$-model-type field theory 
on $\Sigma$. The $\infty$-quasi-isomorphism $i_\infty$
identifying these two $L_\infty$-algebras is obtained from the Homotopy Transfer Theorem
\ref{theo:homotopytransfer}. 

\item On the right-hand side, we have the $L_\infty$-algebra
$\big(\L(X),\ell\big)$ from Definition \ref{def:singularLinfty}, which describes 
$4d$ semi-holomorphic Chern-Simons theory with singularities but no boundary conditions,
and its equivalent model $\big(\L(\Sigma),\ell^\prime\big)$ from 
Proposition \ref{prop:transfersingular}, which as explained in Remark 
\ref{rem:transfersingular} describes a theory of Lax connections. 
The $\infty$-quasi-isomorphism $i_\infty$
identifying these two $L_\infty$-algebras is obtained from the Homotopy Transfer Theorem
\ref{theo:homotopytransfer}, and we denote by $i^{-1}_\infty$ a choice of quasi-inverse,
which exists by \cite[Theorem 10.4.4]{LodayVallette}. 

\item The bottom horizontal arrow is the (strict) $L_\infty$-algebra morphism
which is obtained from the holomorphic line bundle
morphisms $L_{D}\to L_{D^\prime}$ associated with partially ordered divisors $D\leq D^\prime$ 
from Construction \ref{constr:linebundles}. Explicitly, recalling \eqref{eqn:dRDolsingular} 
and \eqref{eqn:dRDolboundary}, the relevant partial orderings are
\begin{flalign}
(\omega)_\infty^0\,\leq \, 0 ~~,\quad
(\omega)_0^\pm + (\omega)_\infty^\pm \,\leq \, (\omega)_0^\pm ~~,\quad
(\omega)_0 + (\omega)_\infty^2 \,\leq \, (\omega)_0 \quad.
\end{flalign}

\item The top horizontal $L_\infty$-morphism is defined by composing the other three arrows. 
Note that it maps from the $L_\infty$-algebra $\big(\F(\Sigma),\ell^\prime\big)$ 
describing a $\sigma$-model-type field theory on $\Sigma$ to the $L_\infty$-algebra 
$\big(\L(\Sigma),\ell^\prime\big)$ describing Lax connections. 
In particular, it maps Maurer-Cartan elements in $\big(\F(\Sigma),\ell^\prime\big)$ 
to Maurer-Cartan elements in $\big(\L(\Sigma),\ell^\prime\big)$, which
according to their explicit descriptions in Remarks \ref{rem:transferboundary} 
and \ref{rem:transfersingular} means that this $L_\infty$-morphism assigns a
Lax connection to each field satisfying the equation of motion.
Note that this is the key feature of an integrable field theory,
so the top horizontal arrow in \eqref{eqn:maindiagram} should be regarded
as the crucial additional structure which renders the field theory $\big(\F(\Sigma),\ell^\prime\big)$ integrable.
\end{itemize}

Summing up our constructions in this and the previous subsections,
we obtain the following main result.
\begin{theo}\label{theo:lax}
The constructions above define an $L_\infty$-morphism
\begin{flalign}
\xymatrix@C=4em{
\big(\F(\Sigma),\ell^\prime\big) \ar@{~>}[r]^-{}~&~\big(\L(\Sigma),\ell^\prime\big) 
}
\end{flalign}
mapping from the $L_\infty$-algebra $\big(\F(\Sigma),\ell^\prime\big)$ describing
a $2d$ integrable field theory on spacetime $\Sigma$ (see Proposition \ref{prop:transferboundary}
and Remark \ref{rem:transferboundary}) to the 
$L_\infty$-algebra $\big(\L(\Sigma),\ell^\prime\big)$ describing
Lax connections (see Proposition \ref{prop:transfersingular} and Remark \ref{rem:transfersingular}).
\end{theo}

\subsection{\label{subsec:cyclic}Cyclic structures and action functionals}
In this subsection we observe that the $L_\infty$-algebra $\big(\F(X),\ell\big)$ 
from Definition \ref{def:boundaryLinfty}, which describes $4d$ semi-holomorphic 
Chern-Simons theory on $X=\Sigma\times C$ with both singularities and boundary conditions,
admits a cyclic structure.
As in the case of ordinary Chern-Simons theory from Example \ref{ex:CSfield},
this cyclic structure depends on the choice of a non-degenerate invariant symmetric bilinear form
$\langle\cdot,\cdot\rangle : \g\otimes\g\to \bbC$ on the Lie algebra $\g$.
We will then show that there exists a homotopy transfer of this cyclic structure
to the $L_\infty$-algebra $\big(\F(\Sigma),\ell^\prime\big)$ from Proposition 
\ref{prop:transferboundary}, which describes a $2d$ $\sigma$-model-type field theory 
on $\Sigma$. 
\sk

To define the cyclic structure on $\big(\F(X),\ell\big)$, 
we use again the holomorphic line bundle
morphisms $L_{D}\to L_{D^\prime}$ associated with partially ordered divisors $D\leq D^\prime$ 
from Construction \ref{constr:linebundles} to define the linear map
\begin{flalign}\label{eqn:wedgepairing}
\wedgepair{\cdot}{\cdot}\,:\,\F(X)\otimes\F(X)~\longrightarrow~
\mathrm{Tot}^{\oplus}\Big(\Omega^{\bullet,(0,\bullet)}\big(X,L_{(\omega)}\big)\Big)~~,\quad
\alpha\otimes\beta ~\longmapsto~\wedgepair{\alpha}{\beta}
\end{flalign}
to the totalized de Rham-Dolbeault bicomplex twisted uniformly by the holomorphic line bundle
$L_{(\omega)}\to C$ associated with the full divisor $(\omega) = (\omega)_0 + (\omega)_\infty$
of the meromorphic $1$-form $\omega$. 
Presenting the meromorphic $1$-form $\omega\in \M^{(1)}(C)\setminus\{0\}$
as a $(1,0)$-form $\omega \in \Omega^{1,0}(C,L_{-(\omega)})$ satisfying
the holomorphicity condition $\delbar \omega =0$, we can define the degree $1$ linear map
\begin{flalign}
\omega\wedge \wedgepair{\cdot}{\cdot}\,:\,\F(X)\otimes\F(X)~\longrightarrow~
\Omega^\bullet(X)~~,\quad
\alpha\otimes\beta ~\longmapsto~\omega\wedge \wedgepair{\alpha}{\beta}
\end{flalign}
to the full de Rham complex of $X=\Sigma\times C$, where we have used that 
$L_{-(\omega)}\otimes L_{(\omega)}\cong C\times \bbC$ is the trivial holomorphic line bundle.
Restricting to compactly supported fields, we can then define via integration on $X$ the degree $-3$ linear map
\begin{flalign}\label{eqn:cyclic}
\pair{\cdot}{\cdot}_{\omega}^{}\,:\,\F_\cc(X)\otimes\F_\cc(X)~\longrightarrow~\bbC~~,\quad
\alpha\otimes\beta ~\longmapsto~\int_X\omega\wedge \wedgepair{\alpha}{\beta}\quad,
\end{flalign}
which satisfies the axioms of a cyclic structure from Remark \ref{rem:cyclic}.
(Note in particular that the property \eqref{eqn:bdycondition2} of a local boundary condition
implies that the integrand is a point-wise non-degenerate pairing.)
In simple words, the reason why \eqref{eqn:cyclic} is a well-defined cyclic structure
is that the singularities of $\omega$ get compensated by the zeros of the fields, and vice versa
that the singularities of the fields get compensated by the zeros of $\omega$. 
This compensation depends of course strongly on the specific assignment
of holomorphic line bundles in the bicomplex \eqref{eqn:dRDolboundary}.
The action functional on fields $\A \in \F_\cc(X)^1$ resulting from this 
cyclic structure reads as
\begin{flalign}
\nn S(\A)\,&:=\,\Big\langle\!\!\Big\langle \A , 
\sum_{n\geq 1} \tfrac{1}{(n+1)!}\,\ell_n\big(\A^{\otimes n}\big) 
\Big\rangle\!\!\Big\rangle_\omega^{}\\
\,&=\, \int_X \omega \wedge \Big(\tfrac{1}{2}\wedgepair{\A}{(\dd_\Sigma + \delbar) \A} + \tfrac{1}{3!}\wedgepair{\A}{\wedgebracket{\A}{\A}}
\Big)\quad,
\end{flalign}
which we recognize as the action for $4d$ semi-holomorphic Chern-Simons theory from \cite{CY3}.
\sk

It is in general \textit{not} true that this cyclic structure 
transfers along the strong deformation retract \eqref{eqn:defretboundary} 
to a cyclic structure on the equivalent $L_\infty$-algebra $\big(\F(\Sigma),\ell^\prime\big)$ 
from Proposition \ref{prop:transferboundary} describing a
$2d$ $\sigma$-model-type field theory on $\Sigma$. It is shown in
\cite{LazarevHodge,LazarevCS} that in order to obtain such a result
the strong deformation retract \eqref{eqn:defretboundary} 
must be compatible with the cyclic structure in the sense that
\begin{subequations}\label{eqn:cycliccompatible}
\begin{flalign}\label{eqn:cycliccompatible1}
\pair{\alpha}{\beta}_{\omega}^{}\,=\,0\quad,
\end{flalign}
for all $\alpha\in\mathrm{im}(\widetilde{i})\subseteq \F_\cc(X)$ 
and $\beta\in \mathrm{ker}(\widetilde{p})\subseteq \F_\cc(X)$,
and
\begin{flalign}\label{eqn:cycliccompatible2}
\pair{h(\alpha)}{\beta}_{\omega}^{} \,=\,(-1)^{\vert\alpha\vert} \,\pair{\alpha}{h(\beta)}_{\omega}^{}\quad,
\end{flalign}
\end{subequations}
for all $\alpha,\beta\in \F_\cc(X)$.
\begin{propo}\label{prop:cyclictransfer}
There exists a choice for the component-wise continuous strong deformation retracts in
\eqref{eqn:defretboundary0} such that the totalized strong deformation retract \eqref{eqn:defretboundary} 
satisfies the compatibility conditions \eqref{eqn:cycliccompatible}.
In particular, with this choice the transferred pairing
\begin{flalign}
\pair{\cdot}{\cdot}_{\omega}^\prime\,:=\,\pair{\cdot}{\cdot}_\omega^{}\circ (\widetilde{i}\otimes\widetilde{i})\,:\, \F_\cc(\Sigma)\otimes \F_\cc(\Sigma)~\longrightarrow~\bbC
\end{flalign}
defines a cyclic structure on the $L_\infty$-algebra $\big(\F(\Sigma),\ell^\prime\big)$ 
from Proposition \ref{prop:transferboundary}.
\end{propo}
\begin{proof}
The construction in Appendix \ref{app:compatible} shows 
that the component-wise continuous strong deformation retracts in
\eqref{eqn:defretboundary0} can be chosen such that 
$(i^0,p^0,h^0)$ and $(i^2,p^2,h^2)$ are compatible, and such that
$(i^+,p^+,h^+)$ and $(i^-,p^-,h^-)$ are compatible. From this it follows 
directly that the \textit{undeformed} totalized strong deformation retract
\begin{equation}
\begin{tikzcd}
\big(\F(\Sigma),0\big)
\ar[r,shift right=-1ex,"i"] & \ar[l,shift right=-1ex,"p"] \big(\F(X),\delbar \big) 
\ar[loop,out=-20,in=20,distance=20,swap,"h"]
\end{tikzcd}
\end{equation}
from \eqref{eqn:undeformedtotalization} satisfies the compatibility 
conditions \eqref{eqn:cycliccompatible}. Hence, it remains to show that
its deformation given in \eqref{eqn:deformedtotal} satisfies the compatibility conditions too.
Note that the second condition \eqref{eqn:cycliccompatible2} is automatic because 
the homotopy $\widetilde{h} = h$ does not receive any deformation, but 
\eqref{eqn:cycliccompatible1} requires a check because both
$\widetilde{i}$ and $\widetilde{p}$ receive deformations. 
Since $\F(\Sigma)$ is concentrated in degrees $\{1,2\}$, there are two cases
depending on whether $\alpha\in \mathrm{im}(\widetilde{i})$ is of degree $1$ or $2$.
\sk

Let us start with the case where $\alpha\in \mathrm{im}(\widetilde{i})^1$ is of degree $1$.
Since $\pair{\cdot}{\cdot}_\omega^{}$ is of degree $-3$, it then suffices to evaluate 
\eqref{eqn:cycliccompatible1} for $\beta \in \mathrm{ker}(\widetilde{p})^2$ of degree $2$.
Such element takes the form
\begin{subequations}
\begin{flalign}
\beta \,=\,\beta^{(2,0,0)} \oplus \beta^{(+,0,1)}\oplus \beta^{(-,0,1)}\,\in\, \F_\cc(X)^2\quad,
\end{flalign}
where the superscript indicates the relevant components in \eqref{eqn:dRDolboundary}, 
and it lies in the kernel of the map $\widetilde{p}$ from \eqref{eqn:deformedtotal}
if and only if 
\begin{flalign}\label{eqn:pkerneltmp}
p^2\Big(\beta^{(2,0,0)} + \dd_\Sigma h^+\big(\beta^{(+,0,1)}\big) + \dd_\Sigma h^-\big(\beta^{(-,0,1)}\big)\Big)\,=\,0\quad.
\end{flalign}
\end{subequations}
Recalling also $\widetilde{i}$ from \eqref{eqn:deformedtotal}, we
can write $\alpha =\widetilde{i}(\Phi) = 
i^0(\Phi) \oplus h^+ \dd_\Sigma^+ i^0(\Phi) \oplus h^-\dd_{\Sigma}^- i^0(\Phi)$, 
for some $\Phi\in \F_\cc(\Sigma)^1$, and compute
\begin{flalign}
\nn \pair{\widetilde{i}(\Phi)}{\beta}_{\omega}^{}\,&=\,
\pair{i^0(\Phi)}{\beta^{(2,0,0)}}_{\omega}^{}+ \pair{ h^+ \dd_\Sigma^+ i^0(\Phi)}{\beta^{(-,0,1)}}_{\omega}^{}
+  \pair{ h^- \dd_\Sigma^- i^0(\Phi)}{\beta^{(+,0,1)}}_{\omega}^{}\\
\,&=\,\pair{i^0(\Phi)}{\beta^{(2,0,0)} + \dd_\Sigma h^+\big(\beta^{(+,0,1)}\big) 
+ \dd_\Sigma h^-\big(\beta^{(-,0,1)}\big)}_{\omega}^{}\,=\,0\quad,
\end{flalign}
where we have used the compatibility conditions for the undeformed totalized strong deformation retract
together with \eqref{eqn:pkerneltmp}.
\sk

Let us consider now the case where $\alpha\in \mathrm{im}(\widetilde{i})^2$ is of degree $2$, 
and hence $\beta \in \mathrm{ker}(\widetilde{p})^1$ is of degree $1$. Since
$\alpha = \widetilde{i}(\Phi^\ddagger) = i^2(\Phi^\ddagger)$ is a form of type $(2,0,0)$, for 
all $\Phi^{\ddagger}\in\F_\cc(\Sigma)^2$, it suffices to consider $\beta$ of type $\beta^{(0,0,1)}$.
The latter lies in the kernel of $\widetilde{p}$ if and only if $p^0\big(\beta^{(0,0,1)}\big)=0$.
We then compute
\begin{flalign}
 \pair{\widetilde{i}(\Phi^\ddagger)}{\beta}_{\omega}^{}\,=\,
 \pair{i^2(\Phi^\ddagger)}{\beta^{(0,0,1)}}_{\omega}^{}\,=\,0
\end{flalign}
by using the compatibility conditions for the undeformed totalized strong deformation retract.
\end{proof}

\begin{rem}
The transferred cyclic structure from Proposition \ref{prop:cyclictransfer}
allows us to define an action functional
\begin{flalign}
S^\prime(\Phi)\,:=\,\Big\langle\!\!\Big\langle \Phi , 
\sum_{n\geq 1} \tfrac{1}{(n+1)!}\,\ell_n^\prime\big(\Phi^{\otimes n}\big) 
\Big\rangle\!\!\Big\rangle_\omega^{\prime}
\end{flalign}
on the fields $\Phi\in\F_\cc(\Sigma)^1=  \Omega_\cc^0\big(\Sigma,\g^{N}\big) $
whose Euler-Lagrange equations are, by construction, the $\sigma$-model-type equations from Remark \ref{rem:transferboundary}.
As a consequence of the integrability structure from Theorem \ref{theo:lax},
this action describes a $2d$ integrable field theory on spacetime $\Sigma$.
\end{rem}


\section{\label{sec:highergenus}The case of higher genus}
In this section we will generalize our results from Section \ref{sec:genuszero}
to the case where the compact Riemann surface $C$ has genus $g > 0$.
It has already been observed in \cite{CY3,Derryberry} that 
$4d$ semi-holomorphic Chern-Simons theory on 
higher genus Riemann surfaces comes with additional subtleties,
which we will try to highlight and explain below from our point of view.
Before diving into the more technical aspects of this section, let 
us start with a general observation: The Definitions \ref{def:auxiliaryLinfty}, 
\ref{def:singularLinfty} and \ref{def:boundaryLinfty} of the $L_\infty$-algebras of $4d$ 
semi-holomorphic Chern-Simons theory on $X = \Sigma\times C$ generalize directly
to the case of higher genus Riemann surfaces, but our analysis
in Propositions \ref{prop:transferauxiliary}, \ref{prop:transfersingular}
and \ref{prop:transferboundary} of their equivalent models obtained by 
computing the $\delbar$-cohomologies uses $g=0$ manifestly.
Hence, in order to explore and understand the similarities and potential differences 
between the cases $g=0$ and $g>0$, one has to revisit these 
cohomology computations in the case of an arbitrary compact Riemann surface $C$.
This will be precisely our starting point for the present section.
\sk

Let us fix an arbitrary compact Riemann surface $C$ 
of genus $g > 0$ and start with the $L_\infty$-algebra $\big(\F(X),\ell\big)$ from
Definition \ref{def:boundaryLinfty} which describes 
$4d$ semi-holomorphic Chern-Simons theory on $X=\Sigma\times C$ with both singularities
and boundary conditions. One of the key features in genus $g=0$ 
which allowed us in Proposition \ref{prop:transferboundary}
and Remark \ref{rem:transferboundary} to identify this $L_\infty$-algebra 
with a $\sigma$-model-type field theory $\big(\F(\Sigma),\ell^\prime\big)$ 
on spacetime $\Sigma$ was the fact that the twisted Dolbeault cohomologies
\eqref{eqn:defretboundary0pm} associated with the $\pm$-form components
along $\Sigma$ vanish in both degrees. This is a special feature
of genus $g=0$, while for arbitrary genus $g\geq 0$ one only has estimates
\begin{flalign}\label{eqn:highergenusestimate}
0\,\leq \,\dim\, \mathsf{H}^0\, \Omega^{0,\bullet}\big(C,L_{(\omega)_0^\pm + (\omega)_{\infty}^\pm}\big)\,=\, 
\dim\,\mathsf{H}^1\, \Omega^{0,\bullet}\big(C,L_{(\omega)_0^\pm + (\omega)_{\infty}^\pm}\big) 
\,\leq\, g\quad.
\end{flalign}
Note that these estimates follow from \cite[Exercise 17.4]{Forster},
the Riemann-Roch Theorem \ref{theo:cohomologies} and the fact that
the divisors $(\omega)_0^\pm + (\omega)_{\infty}^\pm$ have degree $g-1$
by our Definition \ref{def:bdycondition} of local boundary conditions.
While these estimates would in principle allow both cohomologies to vanish,
it is easy to find an explicit example which indicates that this is in general not the case.
\begin{ex}
Consider a complex torus $C = \bbC/\Gamma$, which has genus $g=1$, 
and the meromorphic $1$-form $\omega = \dd z$ without zeros and poles.
In this case the divisor $(\omega)=0$ is trivial, hence
so are the divisors $(\omega)_0^\pm + (\omega)_{\infty}^\pm=0$. 
This reduces the cohomologies in \eqref{eqn:highergenusestimate}
to the untwisted Dolbeault cohomologies, which are non-trivial
$\mathsf{H}^0\, \Omega^{0,\bullet}(C)\cong 
\mathsf{H}^1\, \Omega^{0,\bullet}(C)\cong\bbC$.
\end{ex}

A direct consequence of this issue 
is that the $L_\infty$-algebra $\big(\F(X),\ell\big)$
describing $4d$ semi-holomorphic Chern-Simons theory 
with both singularities and boundary conditions
does \textit{not} in general admit an equivalent description
in terms of a $\sigma$-model-type field theory on $\Sigma$.
This problem was already recognized by Costello and Yamazaki 
in \cite[Section 15.3]{CY3}, who have proposed a 
solution which we will now adapt to our context:
Let $G$ denote a choice of Lie group whose underlying Lie algebra is $\g$.
The undesired cohomologies observed above are a consequence
of the fact that our $L_\infty$-algebra 
$\big(\F(X),\ell\big)$ describes perturbative aspects
(i.e.\ a formal moduli problem) 
of $4d$ semi-holomorphic Chern-Simons theory around
the trivial holomorphic principal $G$-bundle $C\times G\to C$
over the Riemann surface. These cohomologies would vanish 
if one can find a more suitable base point for the perturbative construction,
which is given by a holomorphic principal $G$-bundle $P\to C$
whose adjoint bundle $\g_P\to C$ has the property that 
the divisor-twisted Dolbeault cohomologies
\begin{flalign}\label{eqn:assumption}
\mathsf{H}^0\, \Omega^{0,\bullet}\big(C,\g_P\otimes L_{(\omega)_0^\pm + (\omega)_{\infty}^\pm}\big)
\,\cong\,\mathsf{H}^1\, \Omega^{0,\bullet}\big(C,\g_P\otimes L_{(\omega)_0^\pm + (\omega)_{\infty}^\pm}\big)\,\cong\,0
\end{flalign}
vanish. (Note that the $+$ and $-$ conditions in \eqref{eqn:assumption}
are linked by Serre duality, so it suffices to demand \eqref{eqn:assumption} for either $+$ or $-$.)
Following \cite{CY3} and \cite{Derryberry}, we assume that
such holomorphic principal $G$-bundle $P\to C$ exists and make a choice. Perturbing 
$4d$ semi-holomorphic Chern-Simons theory around $P\to C$ modifies of course the
$L_\infty$-algebras from Definitions \ref{def:auxiliaryLinfty}, 
\ref{def:singularLinfty} and \ref{def:boundaryLinfty}, which is concretely
given by replacing the trivial bundle $C\times \g\to C$ with the corresponding adjoint bundle $\g_P\to C$.
\begin{notation}\label{notation:bundle}
With a slight abuse of notation, we shall denote in what follows the $L_\infty$-algebras which are 
obtained by replacing in Definitions \ref{def:auxiliaryLinfty}, 
\ref{def:singularLinfty} and \ref{def:boundaryLinfty} 
the trivial bundle $C\times \g\to C$ with the adjoint bundle $\g_P\to C$
by the same symbols $\big(\E(X),\ell\big)$,
$\big(\L(X),\ell\big)$ and $\big(\F(X),\ell\big)$.
\end{notation}

\begin{propo}\label{prop:transferboundaryhigher}
Replacing the trivial bundle $C\times \g\to C$ with the adjoint bundle $\g_P\to C$,
the result in Proposition \ref{prop:transferboundary} 
generalizes to any compact Riemann surfaces $C$ with genus $g >0$. In particular,
the $L_\infty$-algebra $\big(\F(X),\ell\big)$ resulting from this replacement 
admits an equivalent model $\big(\F(\Sigma),\ell^\prime\big)$ whose 
underlying cochain complex is of the form \eqref{eqn:defretboundary2}
(with $N := g-1-\deg(\omega)_\infty^0$ extending Notation \ref{not:N})
and whose higher arity structure maps are given by homotopy transfer.
\end{propo}
\begin{proof}
To prove this result, we have to reproduce the cohomology computations
\eqref{eqn:defretboundary0} in our present context. For the
$\pm$-form components along $\Sigma$, we have the twisted Dolbeault complex
$\Omega^{0,\bullet}\big(C,\g_P\otimes L_{(\omega)_0^\pm + (\omega)_{\infty}^\pm}\big)$,
whose cohomology vanishes in both degrees by our hypotheses \eqref{eqn:assumption}.
This directly yields the desired generalization of \eqref{eqn:defretboundary0pm}.
For later use, let us also note that, as a consequence of \eqref{eqn:assumption}, 
$\deg((\omega)_0^\pm + (\omega)_{\infty}^\pm) = g - 1$ and the Riemann-Roch theorem 
\begin{flalign}\label{eqn:RRhigher}
\dim \,\mathsf{H}^0\,\Omega^{0,\bullet}(C,E) - \dim \,\mathsf{H}^1\,\Omega^{0,\bullet}(C,E)\,=\,
c_1(E) + \mathrm{rank}(E) \,(1-g)
\end{flalign}
for higher-rank holomorphic vector bundles $E\to C$, see e.g.\ \cite{Huybrechts}, together with the fact that
$c_1(\g_P \otimes L_D) = c_1(\g_P) + \dim(\g) \, \deg(D)$ for any divisor $D : C \to \mathbb{Z}$,
it follows that the adjoint bundle $\g_P\to C$ must have a trivial first Chern class
$c_1(\g_P)=0$.
\sk

Let us consider now the twisted Dolbeault complex
$\Omega^{0,\bullet}\big(C,\g_P\otimes L_{(\omega)_{\infty}^0}\big)$
for the $0$-form component along $\Sigma$. Observe that by
Definition \ref{def:bdycondition} we have that
\begin{flalign}
(\omega)_{\infty}^0 - \big((\omega)_0^\pm + (\omega)_{\infty}^\pm\big)
\,=\,
\big((\omega)_{\infty}^0 - (\omega)_{\infty}^\pm\big) - (\omega)_0^\pm\,\leq\,0
\end{flalign}
is a non-positive divisor.
It then follows from our hypotheses \eqref{eqn:assumption} and Proposition \ref{prop:vanishing} (a) 
that $\mathsf{H}^0\,\Omega^{0,\bullet}\big(C,\g_P\otimes L_{(\omega)_{\infty}^0}\big)\cong 0$ is trivial.
The Riemann-Roch theorem \eqref{eqn:RRhigher} and $c_1(\g_P)=0$ implies
that $\mathsf{H}^1\,\Omega^{0,\bullet}\big(C,\g_P\otimes L_{(\omega)_{\infty}^0}\big)\cong \g^{N}$
with $N = g-1-\deg(\omega)_\infty^0$, yielding the desired generalization of \eqref{eqn:defretboundary00}.
\sk

Concerning the twisted Dolbeault complex
$\Omega^{0,\bullet}\big(C,\g_P\otimes L_{(\omega)_0+(\omega)_{\infty}^2}\big)$
for the $2$-form component along $\Sigma$, we observe that
\begin{flalign}
(\omega)_0 + (\omega)_\infty^2 - \big((\omega)_0^\pm + (\omega)_{\infty}^\pm\big)\,=\,
(\omega)_0^\mp + \big((\omega)_\infty^2 - (\omega)_{\infty}^\pm\big)\,\geq \,0\quad.
\end{flalign}
From \eqref{eqn:assumption} and Proposition \ref{prop:vanishing} (b) 
it then follows that $\mathsf{H}^1\,\Omega^{0,\bullet}\big(C,\g_P\otimes L_{(\omega)_0+(\omega)_{\infty}^2}\big)\cong 0$ 
is trivial. 
The Riemann-Roch theorem \eqref{eqn:RRhigher} and $c_1(\g_P)=0$ implies
that $\mathsf{H}^0\,\Omega^{0,\bullet}\big(C,\g_P\otimes L_{(\omega)_0+(\omega)_{\infty}^2}\big)\cong \g^{N}$, 
yielding the desired generalization of \eqref{eqn:defretboundary02}.
\end{proof}

\begin{rem}\label{rem:transferboundaryhigher}
The interpretation of $\big(\F(\Sigma),\ell^\prime\big)$
in terms of a $\sigma$-model-type field theory
from Remark \ref{rem:transferboundary} carries over 
verbatim to the present higher genus case. 
Hence, our results are compatible with the ones in 
\cite{CY3,Derryberry}, where it is argued by different
methods that $4d$ semi-holomorphic Chern-Simons theory
on $X=\Sigma\times C$ yields (integrable) $\sigma$-models
on spacetime $\Sigma$ also in the case of higher genus Riemann surfaces $C$.
\end{rem}

Let us consider now the $L_\infty$-algebra $\big(\L(X),\ell\big)$ from
Definition \ref{def:singularLinfty} which describes 
$4d$ semi-holomorphic Chern-Simons theory on $X=\Sigma\times C$ 
with singularities, but no boundary conditions. In genus $g=0$,
we gave in Proposition \ref{prop:transfersingular} and Remark
\ref{rem:transfersingular} an equivalent description 
$\big(\L(\Sigma),\ell^\prime\big)$ of this $L_\infty$-algebra,
which provided a clear interpretation in terms of Lax connections.
We will now investigate how this picture changes when passing to higher genus $g > 0$
and twisting the Dolbeault complexes with the adjoint
bundle $\g_P\to C$ from Notation \ref{notation:bundle}.
\sk

Our approach consists again of describing the relevant twisted Dolbeault cohomologies
which enter the individual components of $\big(\L(X),\ell\big)$.
Let us start with the twisted Dolbeault complexes 
$\Omega^{0,\bullet}\big(C,\g_{P}\otimes L_{(\omega)_0^\pm}\big)$ 
for the $\pm$-form components along $\Sigma$. Observing that
\begin{subequations}
\begin{flalign}
(\omega)_0^\pm - \big((\omega)_0^\pm + (\omega)_{\infty}^\pm\big) \,=\,- (\omega)_{\infty}^\pm \,\geq \,0
\end{flalign}
are non-negative divisors, we conclude from our hypotheses \eqref{eqn:assumption}
and Proposition \ref{prop:vanishing} (b) that the first cohomologies
\begin{flalign}
\mathsf{H}^1\,\Omega^{0,\bullet}\big(C,\g_{P}\otimes L_{(\omega)_0^\pm}\big)\,\cong\, 0
\end{flalign}
\end{subequations}
are trivial. Similarly, for the twisted Dolbeault complex
$\Omega^{0,\bullet}\big(C,\g_{P}\otimes L_{(\omega)_0}\big)$ 
for the $2$-form component along $\Sigma$, we observe that
\begin{subequations}
\begin{flalign}
(\omega)_0 - \big((\omega)_0 + (\omega)_{\infty}^2\big) \,=\,- (\omega)_{\infty}^2 \,\geq \,0\quad,
\end{flalign}
hence, using the result
$\mathsf{H}^1\,\Omega^{0,\bullet}\big(C,\g_P\otimes L_{(\omega)_0+(\omega)_{\infty}^2}\big)\cong 0$
from the proof of Proposition \ref{prop:transferboundaryhigher} combined with Proposition \ref{prop:vanishing} (b),
we deduce that the first cohomology
\begin{flalign}
\mathsf{H}^1\,\Omega^{0,\bullet}\big(C,\g_{P}\otimes L_{(\omega)_0}\big)\,\cong\, 0
\end{flalign}
\end{subequations}
is trivial.
The cohomologies of the twisted Dolbeault complex $\Omega^{0,\bullet}\big(C,\g_{P}\big)$ 
for the $0$-form component along $\Sigma$ seem to be more 
difficult to determine because the adjoint bundle $\g_P\to C$ is specified rather 
indirectly by our cohomological hypotheses in \eqref{eqn:assumption}. Nevertheless, 
a valuable piece of information can be extracted from the Riemann-Roch theorem \eqref{eqn:RRhigher},
which as a consequence of $c_1(\g_P)=0$ implies that
\begin{subequations}
\begin{flalign}
\dim\,\mathsf{H}^0\,\Omega^{0,\bullet}\big(C,\g_{P}\big) - 
\dim\,\mathsf{H}^1\,\Omega^{0,\bullet}\big(C,\g_{P}\big)\,=\, \dim(\g)\,(1-g)\quad.
\end{flalign}
Observing that the right-hand side is negative for genus $g\geq 2$, we obtain that
the first cohomology
\begin{flalign}
\mathsf{H}^1\,\Omega^{0,\bullet}\big(C,\g_{P}\big)\,\not\cong\,0 \qquad \text{(for $g\geq 2$)}
\end{flalign}
\end{subequations}
must be non-trivial for genus $g\geq 2$. (Note that this kind of argument is 
inconclusive in genus $g=1$, where this cohomology may or may not be trivial.)
These cohomologies are a new feature in the case of higher genus,
which in physical terminology correspond to $A_{\bar{z}}$ components
of the Lax connection that cannot be eliminated by gauge transformations.
\sk

To summarize our computations above, let us introduce the short-hand notations 
\begin{flalign}
\O_D(\g_P):= \mathsf{H}^0\,\Omega^{0,\bullet}\big(C,\g_P\otimes L_D\big)\quad,\qquad
\O_D^1(\g_P):= \mathsf{H}^1\,\Omega^{0,\bullet}\big(C,\g_P\otimes L_D\big) \quad,
\end{flalign}
and note that the analogue of \eqref{eqn:defretsingular} in the case of higher genus
is given by the complex
\begin{flalign}\label{eqn:defretsingularhigher}
\begin{gathered}
\big(\L(\Sigma),\dd^\prime\big)\,=\,
\left(\parbox{4cm}{\xymatrix@C=2em@R=0em{
~&~\Omega^{+}(\Sigma)\otimes \O_{(\omega)_0^+}(\g_P)  \ar[rdd]^-{\dd_\Sigma} ~&~ \\
~&~\oplus ~&~ \\
\Omega^{0}(\Sigma)\otimes \O(\g_P) \ar[ruu]^-{\dd_\Sigma^+}\ar[r]^-{\dd_\Sigma^-} \ar[rdd]_-{0}~&~ 
\Omega^{-}(\Sigma)\otimes \O_{(\omega)_0^-}(\g_P) \ar[r]^-{\dd_\Sigma}~&~ 
\Omega^{2}(\Sigma) \otimes \O_{(\omega)_0}(\g_P)\\
~&~\oplus ~&~ \\
~&~\Omega^{0}(\Sigma)\otimes \O^1(\g_P)  \ar[ruu]_-{\dd^\prime} ~&~
}}\right)\quad,
\end{gathered}
\end{flalign}
where the primed differential is given as in \eqref{eqn:defretboundary2} 
by $\dd^\prime := p^2\,\dd_\Sigma\,(h^+\dd_\Sigma^+ + h^-\dd_\Sigma^-)\,i^0$.
This implies that the result in Proposition \ref{prop:transfersingular} 
receives the following modification for a compact Riemann surfaces $C$ of genus $g> 0$.
\begin{propo}\label{prop:transfersingularhigher}
The $L_\infty$-algebra $\big(\L(X),\ell\big)$ obtained by
replacing in Definition \ref{def:singularLinfty}
the trivial bundle $C\times \g\to C$ with the adjoint bundle $\g_P\to C$
admits an equivalent model $\big(\L(\Sigma),\ell^\prime\big)$ whose 
underlying cochain complex is of the form \eqref{eqn:defretsingularhigher}
and whose higher arity structure maps are given by homotopy transfer.
\end{propo}

\begin{rem}\label{rem:transfersingularhigher}
Comparing with the case of genus $g=0$
from Proposition \ref{prop:transfersingular} and 
Remark \ref{rem:transfersingular}, we observe that the main new feature 
is the existence of additional non-vanishing cohomologies 
$\O^1(\g_P)=\mathsf{H}^1\,\Omega^{0,\bullet}\big(C,\g_P\big)  \not\cong 0$,
at least in the case of genus $g\geq 2$. The physical 
interpretation of these cohomology classes is in terms of 
$A_{\bar{z}}$ components of the Lax connection that cannot be 
eliminated by gauge transformations. Maurer-Cartan elements
$A_+ \oplus A_-\oplus\xi\in\L(\Sigma)^1$ in the 
$L_\infty$-algebra $\big(\L(\Sigma),\ell^\prime\big)$ are thus given by
triples of elements
\begin{subequations}
\begin{flalign}
A_+\,\in\, \Omega^{+}(\Sigma)\otimes \O_{(\omega)_0^+}(\g_P) ~~,\quad
A_-\,\in\, \Omega^{-}(\Sigma)\otimes \O_{(\omega)_0^-}(\g_P) ~~,\quad
\xi\,\in\,\Omega^{0}(\Sigma)\otimes \O^1(\g_P)\quad,
\end{flalign}
such that the meromorphic connection $1$-form 
$A = A_+ + A_-\in \Omega^{1}(\Sigma)\widehat{\otimes}\M(\g_P)$
has curvature
\begin{flalign}\label{eqn:Laxnonflat}
\dd_\Sigma A +\tfrac{1}{2}\,\wedgebracket{A}{A}\,=\, - \dd^\prime\xi + \cdots\quad
\end{flalign}
\end{subequations}
which is determined by the new cohomology classes $\xi$. The
terms omitted on the right-hand side are obtained from the higher arity components
$\ell^\prime_n$, $n\geq 2$, of the transferred $L_\infty$-structure.
Hence, in the model provided by the $L_\infty$-algebra $\big(\L(\Sigma),\ell^\prime\big)$,
Lax connections on higher genus Riemann surfaces are meromorphic objects, but they have in general
a non-trivial curvature which is specified by the Maurer-Cartan equation \eqref{eqn:Laxnonflat}. 
\sk

For readers who are irritated by such not necessarily flat Lax connections, 
let us go back to the equivalent original model $\big(\L(X),\ell\big)$ 
for the $L_\infty$-algebra of Lax connections and recover flatness 
at the price of meromorphicity. 
Indeed, Maurer-Cartan elements in $\big(\L(X),\ell\big)$ 
are given by triples $B_+ \oplus B_-\oplus \zeta \in \L(X)^1$
satisfying the equations
\begin{flalign}
\dd_\Sigma \big(B_+ \oplus B_-\big) + \tfrac{1}{2}\,\wedgebracket{B_+\oplus B_-}{B_+\oplus B_-} \,=\,0\quad,\qquad
\delbar B_{\pm} + \dd_\Sigma^\pm \zeta + \wedgebracket{B_{\pm}}{\zeta}\,=\,0\quad.
\end{flalign}
The first equation manifestly encodes flatness of the connection $B_+ \oplus B_-$ along $\Sigma$,
while the second equation determines a violation of the meromorphicity of $B_\pm$ 
in terms of the field $\zeta$. Hence, in this alternative but equivalent model,
Lax connections on higher genus Riemann surfaces are flat objects, but they are not necessarily 
meromorphic since the field $\zeta$ cannot in general be eliminated by gauge transformations 
due to the non-vanishing cohomologies 
$\mathsf{H}^1\,\Omega^{0,\bullet}\big(C,\g_P\big)  \not\cong 0$.
\sk

To summarize, we presented two equivalent models $\big(\L(X),\ell\big)$ and 
$\big(\L(\Sigma),\ell^\prime\big)$ for the $L_\infty$-algebra 
of Lax connections on higher genus Riemann surfaces. 
From the point of view of the first model, Lax connections are flat objects 
whose meromorphicity is obstructed, while from the point of view of the second 
model, they are meromorphic objects whose flatness is obstructed. 
In both cases the obstruction comes from the non-vanishing cohomologies 
$\mathsf{H}^1\,\Omega^{0,\bullet}\big(C,\g_P\big)  \not\cong 0$.
\end{rem}

\begin{rem}\label{rem:holonomy}
The feature from Remark \ref{rem:transfersingularhigher} 
that either meromorphicity or flatness of the Lax connection 
is obstructed in the higher genus case might look at first sight 
worrying from an integrability perspective. Fortunately, it turns 
out that this does \textit{not} obstruct the 
usual construction of conserved charges from Wilson 
loops of the Lax connection.\footnote{We are grateful to Sylvain Lacroix for pointing this out to us.}
We now explain why this is the case, assuming that the underlying
Lie group $G$ is a matrix Lie group in order to simplify our presentation. To build conserved
charges from Wilson loops, let us consider a sufficiently small open subset
$U\subset C$ of the Riemann surface on which the holomorphic vector bundle $\g_P\to C$ trivializes 
and the divisor $(\omega)\vert_U = 0$ is trivial.
Working in the equivalent model $\big(\L(X),\ell\big)$ in which Lax connections
are flat but not meromorphic, we can identify
locally on $U$ a Lax connection with differential forms $B\in \Omega^{1,(0,0)}(\Sigma\times U,\g)$
and $\zeta \in \Omega^{0,(0,1)}(\Sigma\times U,\g)$ satisfying
\begin{flalign}\label{eqn:Laxhighergenus}
\dd_\Sigma B +\tfrac{1}{2}\,\wedgebracket{B}{B} \,=\,0\quad,\qquad
\delbar B + \dd_\Sigma \zeta + \wedgebracket{B}{\zeta}\,=\,0\quad.
\end{flalign}
Recall that, given a smooth curve $\gamma:\bbR\to\Sigma\,,~s\mapsto\gamma(s)$ in spacetime, 
the holonomy $h_\gamma : \bbR\times U \to G$ of $B$ is determined by
the differential equation
\begin{flalign}
\dd_\bbR h_\gamma + \gamma^\ast(B)\,h_\gamma\,=\,0\quad,\qquad h_{\gamma}\vert_{s=0}\,=\,\oone\quad.
\end{flalign}
As a consequence of the flatness \eqref{eqn:Laxhighergenus} of $B$,
one obtains for a loop $\gamma$ with $\gamma(0) = \gamma(1)$ a family of 
conserved charges $\mathrm{Tr}\big(h_\gamma\vert_{s=1}\big)$ 
on $\Sigma$ which depends on the complex parameter $U\subset C$. We now show that 
$\mathrm{Tr}\big(h_\gamma\vert_{s=1}\big)$ depends holomorphically on $U\subset C$.
The key observation is that the quantity $h_\gamma^{-1} \delbar h_\gamma$ satisfies the differential
equation
\begin{flalign}
\dd_\bbR\big(h_\gamma^{-1} \delbar h_\gamma\big)\,=\, h_{\gamma}^{-1}\,\gamma^\ast(\delbar B)\, h_\gamma
\,=\, -\dd_{\bbR}\big(h_{\gamma}^{-1}\,\gamma^\ast(\zeta)\,h_{\gamma}\big)\quad,\qquad
(h_\gamma^{-1}\, \delbar h_\gamma)\vert_{s=0}\,=\,0\quad,
\end{flalign}
where in the last step we used the second equation in \eqref{eqn:Laxhighergenus}.
The unique solution is
\begin{flalign}
h_\gamma^{-1} \delbar h_\gamma\, =\, \gamma^\ast(\zeta)\vert_{s=0} -h_{\gamma}^{-1} \,\gamma^\ast(\zeta)\, h_\gamma\quad
\Longleftrightarrow\quad
\delbar h_\gamma \,=\, h_\gamma \,\gamma^\ast(\zeta)\vert_{s=0}  - \gamma^\ast(\zeta)\, h_\gamma
\quad.
\end{flalign}
For a loop with $\gamma(0) = \gamma(1)$ one has $\gamma^\ast(\zeta)\vert_{s=1} = \gamma^\ast(\zeta)\vert_{s=0}$,
hence 
\begin{flalign}
\delbar \mathrm{Tr}\big(h_\gamma\vert_{s=1}\big)\,=\, \mathrm{Tr}\Big(
h_\gamma\vert_{s=1} \,\gamma^\ast(\zeta)\vert_{s=0}  - \gamma^\ast(\zeta)\vert_{s=0}\, h_\gamma\vert_{s=1}\Big)\,=\,0
\end{flalign}
by using cyclicity of the trace.
\end{rem}


\section*{Acknowledgments}
We would like to thank Johannes Hofscheier and Sylvain Lacroix
for useful discussions. The work of M.B.\ is supported in part by the MUR Excellence 
Department Project awarded to Dipartimento di Matematica, 
Universit{\`a} di Genova (CUP D33C23001110001) and it is fostered by 
the National Group of Mathematical Physics (GNFM-INdAM (IT)). 
A.S.\ and B.V.\ gratefully acknowledge the support of the Engineering and Physical
Sciences Research Council (UKRI1723).


\appendix

\section{\label{app:details}Technical details}

\subsection{\label{app:totaldefretract}Totalization of strong deformation retracts}
A recurring problem in the bulk of our paper is as follows: 
Given a bicomplex $(V^{\bullet,\bullet},\dd,\delbar)$ 
and a family of strong deformation retracts
\begin{equation}\label{eqn:componentwiseretracts}
\begin{tikzcd}
\big(V^{\prime k,\bullet},\delbar^\prime\big)
\ar[r,shift right=-1ex,"i^k"] & \ar[l,shift right=-1ex,"p^k"] \big(V^{k,\bullet},\delbar\big) 
\ar[loop,out=-20,in=20,distance=20,swap,"h^k"]
\end{tikzcd}\quad,~~\text{for all } k\in\bbZ\quad,
\end{equation}
we would like to construct a strong deformation retract from the totalization
$\big(\mathrm{Tot}^{\oplus}\big(V^{\bullet,\bullet}\big),\dd+\delbar\big)$.
In order to match the scenarios arising in our paper, we can assume
that the given bicomplex is concentrated in vertical degrees $0$ and $1$,
i.e.\ $V^{p,q} = 0$ for all $p\in\bbZ$ and all $q\not\in \{0,1\}$.
\sk

Ignoring for the moment the differential $\dd$, we obtain a strong deformation retract
\begin{equation}\label{eqn:undeformedtotalization}
\begin{tikzcd}
\big(\mathrm{Tot}^{\oplus}\big(V^{\prime \bullet,\bullet}\big),\delbar^\prime\big)
\ar[r,shift right=-1ex,"i"] & \ar[l,shift right=-1ex,"p"] \big(\mathrm{Tot}^{\oplus}\big(V^{\bullet,\bullet}\big),\delbar\big) 
\ar[loop,out=-20,in=20,distance=30,swap,"h"]
\end{tikzcd}
\end{equation}
with the linear maps $i,p,h$ defined component-wise on the degree $(k,\bullet)$
elements of the totalizations by $i^k,p^k,h^k$, for all $k\in\bbZ$. 
Observe that the differential $\dd$ defines a perturbation of the cochain complex
$\big(\mathrm{Tot}^{\oplus}\big(V^{\bullet,\bullet}\big),\delbar\big)$ since
$(\dd+\delbar)^2=0$ as a consequence of $(V^{\bullet,\bullet},\dd,\delbar)$ being a bicomplex.
Note further that this perturbation is small: Since the linear map $\dd\,h : 
\mathrm{Tot}^{\oplus}\big(V^{\bullet,\bullet}\big)\to \mathrm{Tot}^{\oplus}\big(V^{\bullet,\bullet}\big)$ 
is of bidegree $(1,-1)$ and $V^{\bullet,\bullet}$ is by hypothesis
concentrated in vertical degrees $\{0,1\}$, it follows for degree reasons 
that $(\dd\,h)^2 =0$ and hence $\id-\dd\,h$ is invertible by
\begin{flalign}
(\id-\dd\,h)^{-1} \,=\, \id + \dd\,h \,:\,\mathrm{Tot}^{\oplus}\big(V^{\bullet,\bullet}\big)~\longrightarrow~
\mathrm{Tot}^{\oplus}\big(V^{\bullet,\bullet}\big)\quad.
\end{flalign}
Applying the homological perturbation lemma from Proposition \ref{prop:HPT}
then yields the deformed strong deformation retract
\begin{subequations}\label{eqn:deformedtotal}
\begin{equation}
\begin{tikzcd}
\big(\mathrm{Tot}^{\oplus}\big(V^{\prime \bullet,\bullet}\big),\dd^\prime + \delbar^\prime\big)
\ar[r,shift right=-1ex,"\widetilde{i}"] & \ar[l,shift right=-1ex,"\widetilde{p}"] \big(\mathrm{Tot}^{\oplus}\big(V^{\bullet,\bullet}\big),
\dd + \delbar\big) 
\ar[loop,out=-20,in=20,distance=40,swap,"\widetilde{h}"]
\end{tikzcd}
\end{equation}
with
\begin{flalign}
\dd^\prime\,&=\,p\,\dd\,i + p\,\dd\,h\,\dd\,i\quad,\\
\widetilde{i}\,&=\,i + h\,\dd\,i\quad,\\
\widetilde{p}\,&=\,p + p\,\dd\,h\quad,\\
\widetilde{h}\,&=\,h\quad.
\end{flalign}
\end{subequations}
\begin{rem}\label{rem:deformedtotal}
In some (but not all) of the models studied in this work,
the deformed strong  deformation retract \eqref{eqn:deformedtotal}
admits a further simplification.
Suppose that the differential $\dd : V^{\bullet,\bullet}\to V^{\bullet,\bullet}$, describing the
perturbation, restricts along the (necessarily injective) map
$i : V^{\prime\bullet,\bullet}\to V^{\bullet,\bullet}$ obtained from the
component-wise strong deformation retracts \eqref{eqn:componentwiseretracts}, i.e.\ 
\begin{flalign}
\begin{gathered}
\xymatrix@C=3em{
\ar[d]_-{i} V^{\prime\bullet,\bullet} \ar@{-->}[r]^-{\dd}~&~  V^{\prime\bullet,\bullet}\ar[d]^-{i}\\
 V^{\bullet,\bullet} \ar[r]_-{\dd}~&~ V^{\bullet,\bullet}
}
\end{gathered}\quad.
\end{flalign}
From the side condition $h\,i=0$ for a strong deformation retract, 
it then follows that \eqref{eqn:deformedtotal} simplifies to 
\begin{flalign}
\dd^\prime\,=\,\dd ~~,\quad
\widetilde{i}\,=\,i~~, \quad
\widetilde{p}\,=\,p + p\,\dd\,h~~,\quad
\widetilde{h}\,=\,h\quad.
\end{flalign}
In particular, note that in this case the cochain complex on 
the left-hand side of \eqref{eqn:deformedtotal} agrees with the
totalization of the bicomplex
$(V^{\prime \bullet,\bullet},\dd,\delbar^\prime)$ 
and the inclusion map $\widetilde{i}=i$ does 
not receive any deformations.
\end{rem}

\subsection{\label{app:Hodgetheory}Hodge theory}
In this appendix we briefly recall some basic aspects 
of Hodge theory for elliptic complexes, 
see e.g.\ \cite{Warner,Wells} for introductions. 
For the purpose of our paper, it suffices
to consider the case of the Dolbeault complex 
\begin{flalign}
\big(\Omega^{0,\bullet}(C,E),\delbar\big)\,=\,
\Big(
\xymatrix{
\Omega^{0,0}(C,E)\ar[r]^-{\delbar}~&~ \Omega^{0,1}(C,E)
}
\Big)
\end{flalign}
on a compact Riemann surface $C$ which is twisted by a holomorphic vector bundle $E\to C$.
\sk

Recall that for any choice of Hermitian metric on $C$ and Hermitian fiber metric $h$ on $E\to C$, 
one can define a (complex-antilinear) Hodge operator
\begin{flalign}
\ast_h\,:\, \xymatrix{
\overline{\Omega^{p,q}(C,E)} \ar[r]^-{\cong}~&~\Omega^{1-p,1-q}(C,E^\ast)
}\quad,
\end{flalign}
for all $p,q\in\{0,1\}$, where $\overline{V}$ denotes the complex conjugate of a vector space $V$
and $E^\ast\to C$ denotes the dual of the holomorphic vector bundle $E\to C$.
We denote the associated Hermitian inner products by
\begin{flalign}
\pair{\alpha}{\beta} \,:=\,\int_C \ast_h(\alpha)\wedge \beta\quad,
\end{flalign}
for all $\alpha,\beta\in \Omega^{p,q}(C,E)$.
The adjoint $\delbar^\ast : \Omega^{p,q}(C,E)\to \Omega^{p,q-1}(C,E)$ of 
the $\delbar$-operator is defined by $\pair{\delbar^\ast\alpha}{\beta}=\pair{\alpha}{\delbar \beta}$, 
for all $\alpha\in \Omega^{p,q}(C,E)$ and $\beta\in \Omega^{p,q-1}(C,E)$, and it reads explicitly as
\begin{flalign}\label{eqn:delbaradjointexplicit}
\delbar^{\ast} \alpha\,=\,- (-1)^{p+q} \,\ast_h^{-1}\delbar \ast_h\!(\alpha)\quad,
\end{flalign} 
for all $\alpha\in \Omega^{p,q}(C,E)$. We denote by 
$\Delta := \delbar^\ast \,\delbar + \delbar\,\delbar^\ast: \Omega^{p,q}(C,E)\to \Omega^{p,q}(C,E)$
the corresponding Laplacians.
\sk

Specializing the Hodge decomposition theorem to 
the twisted Dolbeault complex $\big(\Omega^{0,\bullet}(C,E),\delbar\big)$
then gives orthogonal direct sum decompositions
\begin{flalign}
\Omega^{0,0}(C,E)\,=\,\mathrm{ker}(\delbar)\oplus \mathrm{im}(\delbar^\ast)\quad,\qquad
\Omega^{0,1}(C,E)\,=\,\mathrm{ker}(\delbar^\ast)\oplus \mathrm{im}(\delbar)\quad,
\end{flalign}
such that the twisted Dolbeault cohomologies can be identified as
\begin{flalign}
\mathsf{H}^0\,\Omega^{0,\bullet}(C,E)\,=\,\mathrm{ker}(\delbar)\quad,\qquad
\mathsf{H}^1\,\Omega^{0,\bullet}(C,E)\,\cong\,\mathrm{ker}(\delbar^\ast) \quad.
\end{flalign}
We denote by $i : \mathsf{H}^\bullet\,\Omega^{0,\bullet}(C,E)\to \Omega^{0,\bullet}(C,E)$
and $p: \Omega^{0,\bullet}(C,E)\to \mathsf{H}^\bullet\,\Omega^{0,\bullet}(C,E)$ the inclusion
and projection maps associated with these direct sum decompositions.
Hodge theory further implies the existence of Green's operators
$G : \Omega^{0,q}(C,E)\to \Omega^{0,q}(C,E)$ satisfying $\Delta\,G = G\,\Delta = \id - i\,p$,
$\delbar\,G = G\,\delbar$ and $\delbar^\ast\,G = G\,\delbar^\ast$. This implies that
\begin{equation}
\begin{tikzcd}
\big(\mathsf{H}^\bullet\,\Omega^{0,\bullet}(C,E),0\big) \ar[r,shift right=-1ex,"i"] & \ar[l,shift right=-1ex,"p"] \big(\Omega^{0,\bullet}(C,E),\delbar \big) \ar[loop,out=-20,in=20,distance=30,swap,"h:=-\delbar^\ast G"]
\end{tikzcd}
\end{equation}
defines a strong deformation retract from the twisted Dolbeault complex 
$\big(\Omega^{0,\bullet}(C,E),\delbar \big)$ to its cohomology.
Using some basic results from the theory of Fr{\'e}chet spaces,
one easily argues that this strong  deformation retract is continuous with respect to 
the finite-dimensional vector space topologies on $\mathsf{H}^\bullet\,\Omega^{0,\bullet}(C,E)$ 
and the standard Fr{\'e}chet topology on $\Omega^{0,\bullet}(C,E)$.
These arguments are spelled out explicitly in \cite[Example 5.6]{BGMS}.

\subsection{\label{app:compatible}Compatible pairs of continuous strong deformation retracts}
To obtain a homotopy transfer of cyclic structures in Subsection \ref{subsec:cyclic},
one has to be able to choose continuous strong deformation retracts which are compatible
with the relevant cyclic structure that is built from a choice of meromorphic 
$1$-form $\omega\in \M^{(1)}(C)\setminus\{0\}$. In this appendix we investigate and
solve this problem at the level of its building blocks, which allows us to streamline the
proof of Proposition \ref{prop:cyclictransfer}.
\sk

Let us start with stating the problem precisely: Suppose that we are given a 
meromorphic $1$-form $\omega\in \M^{(1)}(C)\setminus\{0\}$, a holomorphic vector bundle
$E\to C$ and a divisor $D:C\to \bbZ$ on a compact Riemann surface $C$. Presenting $\omega$
as a $(1,0)$-form $\omega\in\Omega^{1,0}(C,L_{-(\omega)})$ satisfying
$\delbar\omega=0$ and denoting by 
\begin{flalign}
D^\prime \,:=\, (\omega)-D
\end{flalign}
the complementary divisor, we obtain the non-degenerate pairing
\begin{flalign}
\pair{\cdot}{\cdot}_\omega^{}\,:\, \Omega^{0,\bullet}(C,E^\ast\otimes L_{D^\prime})\otimes
 \Omega^{0,\bullet}(C,E\otimes L_D)~\longrightarrow~\bbC~~,\quad
\zeta\otimes\alpha~\longmapsto~\int_C\omega\wedge\zeta\wedge\alpha
\end{flalign}
of degree $-1$. Our goal is to show that there exist 
two continuous strong deformation retracts
\begin{subequations}
\begin{equation}
\begin{tikzcd}
\big(\mathsf{H}^\bullet\,\Omega^{0,\bullet}(C,E\otimes L_D),0\big) \ar[r,shift right=-1ex,"i_D"] & \ar[l,shift right=-1ex,"p_D"] \big(\Omega^{0,\bullet}(C,E\otimes L_D),\delbar \big) \ar[loop,out=-20,in=20,distance=45,swap,"h_D"]
\end{tikzcd}
\end{equation}
and
\begin{equation}
\begin{tikzcd}
\big(\mathsf{H}^\bullet\,\Omega^{0,\bullet}(C,E^\ast\otimes L_{D^\prime}),0\big) \ar[r,shift right=-1ex,"i_{D^\prime}"] & \ar[l,shift right=-1ex,"p_{D^\prime}"] \big(\Omega^{0,\bullet}(C,E^\ast\otimes L_{D^\prime}),\delbar \big) \ar[loop,out=-20,in=20,distance=45,swap,"h_{D^\prime}"]
\end{tikzcd}\quad,
\end{equation}
\end{subequations}
which are compatible with $\pair{\cdot}{\cdot}_\omega^{}$ in the following sense:
\begin{enumerate}
\item $\pair{\zeta}{\alpha}_\omega^{} =0$, 
for all $\zeta\in \mathrm{im}(i_{D^\prime})$ and $\alpha \in \mathrm{ker}(p_{D})$,

\item $\pair{\zeta}{\alpha}_\omega^{} =0$, 
for all $\zeta\in \mathrm{ker}(p_{D^\prime})$ and $\alpha \in \mathrm{im}(i_{D})$, and

\item $\pair{h_{D^\prime}(\zeta)}{\alpha}_\omega^{} = 
(-1)^{\vert \zeta\vert}\, \pair{\zeta}{h_D(\alpha)}_\omega^{}$, for all 
$\zeta\in \Omega^{0,\bullet}(C,E^\ast\otimes L_{D^\prime})$ and $\alpha \in \Omega^{0,\bullet}(C,E\otimes L_D)$.
\end{enumerate}

To obtain the first continuous strong deformation retract $(i_D,p_D,h_D)$, we apply Hodge theory
from Appendix \ref{app:Hodgetheory}. Recall that this requires the choice of a Hermitian metric on $C$
and a Hermitian fiber metric on $E\otimes L_D\to C$, determining the (complex-antilinear) Hodge operator
$\ast_h :\overline{\Omega^{p,q}(C,E\otimes L_D)}\to\Omega^{1-p,1-q}(C,(E \otimes L_D)^\ast) \cong 
\Omega^{1-p,1-q}(C,E^\ast\otimes L_{-D})$.
In order to keep track of the divisor, we 
denote the associated Hermitian inner product, adjoint differential, Laplacian and Green's operator 
by $\pair{\cdot}{\cdot}_{D}^{}$, $\delbar_D^\ast$, $\Delta_{D}$ and $G_D$.
We then define a Hermitian inner product on the complementary twisted Dolbeault complex
$\Omega^{0,\bullet}(C,E^\ast\otimes L_{D^\prime})$ by
\begin{flalign}
\nn \pair{\cdot}{\cdot}_{D^\prime}^{}\,:\,\overline{\Omega^{0,q}(C,E^\ast\otimes L_{D^\prime})}
\otimes \Omega^{0,q}(C,E^\ast \otimes L_{D^\prime})
~&\longrightarrow~\bbC~~,\quad\\
\zeta\otimes \eta~&\longmapsto~\int_{C}\ast_{h}^{-1}\big(\omega\wedge\zeta\big)\wedge(\omega\wedge\eta)\quad,
\end{flalign}
which is well-defined as a consequence of 
$L_{-(\omega)}\otimes L_{D^\prime} = L_{-(\omega)}\otimes L_{(\omega)-D}\cong L_{-D}$.
The corresponding adjoint $\delbar_{D^\prime}^\ast$ of the $\delbar$-operator 
is defined by $\pair{\delbar_{D^\prime}^\ast\zeta}{\eta}_{D^\prime}^{}=\pair{\zeta}{\delbar \eta}_{D^\prime}^{}$, 
for all $\zeta\in \Omega^{0,1}(C,E^\ast\otimes L_{D^\prime})$ and $\eta\in \Omega^{0,0}(C,E^\ast\otimes L_{D^\prime})$, 
and it reads explicitly as
\begin{flalign}\label{eqn:delbaradjointexplicitwithomega}
\omega\wedge \delbar^{\ast}_{D^\prime} \zeta\,=\, \ast_h\delbar \ast_h^{-1}\!\big(\omega\wedge\zeta)\quad,
\end{flalign} 
for all $\zeta\in \Omega^{0,1}(C,E^\ast\otimes L_{D^\prime})$. We denote by 
$\Delta_{D^\prime} := \delbar_{D^\prime}^\ast \,\delbar + \delbar\,\delbar_{D^\prime}^\ast$ 
the corresponding Laplacian. Applying Hodge theory to this context yields the second
continuous strong deformation retract $(i_{D^\prime},p_{D^\prime},h_{D^\prime})$
with $h_{D^\prime} = - \delbar^\ast_{D^\prime} G_{D^\prime}$ determined by
the Green's operator $G_{D^\prime}$ for $\Delta_{D^\prime}$.
\sk

The key ingredient which allows one to prove that these 
two continuous strong deformation retracts are compatible
with $\pair{\cdot}{\cdot}_\omega^{}$ is the observation that 
the adjoint differentials satisfy
\begin{flalign}\label{eqn:delbarastcompatible}
\pair{\delbar_{D^\prime}^\ast\zeta }{\alpha}_\omega^{}\,=\,-(-1)^{\vert\zeta\vert}\,\pair{\zeta}{\delbar^\ast_{D}\alpha}_\omega^{}\,=\,\pair{\zeta}{\delbar^\ast_{D}\alpha}_\omega^{}\quad,
\end{flalign}
for all $\zeta\in \Omega^{0,1}(C,E^\ast\otimes L_{D^\prime})$ and $\alpha\in \Omega^{0,1}(C,E\otimes L_D)$. 
This can be shown by a short calculation using the explicit formulas \eqref{eqn:delbaradjointexplicit}
for $\delbar_D^\ast$ and \eqref{eqn:delbaradjointexplicitwithomega} for $\delbar^{\ast}_{D^\prime}$.
The first compatibility condition then follows from the observation that the pairing
$\pair{\cdot}{\cdot}_\omega^{}$ acts trivially on
$\mathrm{im}(i_{D^\prime})^{0} = \mathrm{ker}(\delbar)$ and $\mathrm{ker}(p_{D})^1 = \mathrm{im}(\delbar)$,
as well as on $\mathrm{im}(i_{D^\prime})^{1} = \mathrm{ker}(\delbar_{D^\prime}^{\ast})$ and 
$\mathrm{ker}(p_{D})^0 = \mathrm{im}(\delbar_D^\ast)$. 
Similarly, the second compatibility condition follows from the observation that
the pairing $\pair{\cdot}{\cdot}_\omega^{}$ acts trivially on
$\mathrm{ker}(p_{D^\prime})^0 = \mathrm{im}(\delbar_{D^\prime}^\ast)$
and $\mathrm{im}(i_{D})^{1} = \mathrm{ker}(\delbar_{D}^{\ast})$,
as well as on $\mathrm{ker}(p_{D^\prime})^1 = \mathrm{im}(\delbar)$ and 
$\mathrm{im}(i_{D})^{0} = \mathrm{ker}(\delbar)$.
\sk

To prove the third compatibility condition, note that as a consequence of 
\eqref{eqn:delbarastcompatible} the Laplacian satisfies $\pair{\Delta_{D^\prime}\zeta}{\alpha}_{\omega}^{} 
= -\pair{\zeta}{\Delta_D \alpha}_{\omega}^{}$. We then compute,
for all  $\zeta\in \Omega^{0,1}(C,E^\ast\otimes L_{D^\prime})$ and $\alpha\in \Omega^{0,1}(C,E\otimes L_{D})$,
\begin{flalign}
\nn \pair{h_{D^\prime}(\zeta)}{\alpha}_\omega^{}\,&=\,
-\pair{\delbar^\ast_{D^\prime}G_{D^\prime}(\zeta)}{\alpha}_\omega^{}
\,=\,-\pair{\delbar^\ast_{D^\prime}G_{D^\prime}(\zeta)}{\Delta_D G_D(\alpha)}_\omega^{}\\
\nn \,&=\,\pair{\Delta_{D^\prime}\delbar^\ast_{D^\prime}G_{D^\prime}(\zeta)}{G_D(\alpha)}_\omega^{}\,=\,
\pair{\delbar^\ast_{D^\prime}\zeta}{G_D(\alpha)}_\omega^{}\\
\,&=\,-(-1)^{\vert \zeta\vert} \, \pair{\zeta}{\delbar^\ast_{D}G_D(\alpha)}_\omega^{}
\,=\, (-1)^{\vert \zeta\vert} \, \pair{\zeta}{h_D(\alpha)}_\omega^{}\quad.
\end{flalign}
In the second step we have used that $\delbar^\ast_{D^\prime}G_{D^\prime}(\zeta)\in \mathrm{ker}(p_{D^\prime})^0$
pairs trivially with the second term of $\alpha = \Delta_D G_D(\alpha) + i_D p_D(\alpha)$ as a consequence
of the second compatibility condition above. In the fourth step we have used
$\Delta_{D^\prime}\delbar^\ast_{D^\prime} = \delbar^\ast_{D^\prime} \Delta_{D^\prime}$ and the fact
that the second term of $\Delta_{D^\prime}G_{D^\prime}(\zeta) = \zeta - i_{D^\prime}p_{D^\prime}(\zeta)$
lies in the kernel of  $\delbar^\ast_{D^\prime}$. This concludes our construction of a compatible
pair of continuous strong deformation retracts.

\subsection{\label{app:vanishing}Relative vanishing results for holomorphic vector bundles}
The following results are needed for carrying out our cohomological 
analysis in the higher genus case in Section \ref{sec:highergenus}.
We believe that such results are standard in complex geometry,
but we were not able to locate them in the literature. Hence,
we will provide elementary proofs.
\begin{propo}\label{prop:vanishing}
Let $C$ be a compact Riemann surface of genus $g\geq 0$
and $E\to C$ a holomorphic vector bundle (possibly of higher rank).
Let further $D:C\to \bbZ$ be a divisor and denote by 
$L_D\to C$ its associated holomorphic line bundle from Construction \ref{constr:linebundles}.
\begin{itemize}
\item[(a)] Suppose that $\mathsf{H}^0\,\Omega^{0,\bullet}(C,E)\cong 0$ is trivial
and $D\leq 0$ is non-positive. Then
\begin{flalign}
\mathsf{H}^0\,\Omega^{0,\bullet}\big(C,E\otimes L_D\big)\,\cong\,0
\end{flalign}
is trivial too.

\item[(b)] Suppose that $\mathsf{H}^1\,\Omega^{0,\bullet}(C,E)\cong 0$ is trivial
and $0\leq D$ is non-negative. Then
\begin{flalign}
\mathsf{H}^1\,\Omega^{0,\bullet}\big(C,E\otimes L_D\big)\,\cong\,0
\end{flalign}
is trivial too.
\end{itemize}
\end{propo}
\begin{proof}
Let us start with observing that (b) follows from (a) and Serre duality: The hypothesis
$\mathsf{H}^1\,\Omega^{0,\bullet}(C,E)\cong 0$ implies via Serre duality that
$\mathsf{H}^0\,\Omega^{0,\bullet}(C,E^\ast\otimes K_C)\cong 0$, where $K_C\to C$ 
denotes the canonical holomorphic line bundle. Item (a) then implies that
$\mathsf{H}^0\,\Omega^{0,\bullet}(C,E^\ast\otimes K_C\otimes L_{-D})\cong 0$
since $-D\leq 0$ is non-positive. Using $L_{-D}^\ast\cong L_D$ and
Serre duality once more then gives $\mathsf{H}^1\,\Omega^{0,\bullet}(C,E\otimes L_{D})\cong 0$.
\sk

To prove item (a), we show that $0: C \to E\otimes L_D$ is the only 
holomorphic section of $E\otimes L_D \to C$. 
Using the \v{C}ech description for $L_D\to C$ from Construction
\ref{constr:linebundles}, a holomorphic section
$s : C\to E\otimes L_D$ is specified by a family
$\big(s_i\in \O(E\vert_{U_i})\big)_{i\in\mathcal{I}}$
of local holomorphic sections which satisfy $s_i\vert_{U_{ij}} = g_{ij}\,s_{j}\vert_{U_{ij}}$
on all overlaps $U_{ij}$ for the transition functions $g_{ij} = \psi_i\vert_{U_{ij}}/\psi_{j}\vert_{U_{ij}}$
of the bundle $L_D\to C$. Since $D\leq 0$ is non-positive, the meromorphic functions
$\psi_i\in\M(U_i)$ representing the divisor $D\vert_{U_i} =(\psi_i)$
only have poles but no zeros, hence $1/\psi_i\in \O(U_i)$ is holomorphic for all $i\in\mathcal{I}$.
From this it follows that $\big(\tilde{s}_i:=s_i/\psi_i\in \O(E\vert_{U_i})\big)_{i\in\mathcal{I}}$
defines a holomorphic section $\tilde{s} : C\to E$, which as a consequence of
our hypothesis $\mathsf{H}^0\,\Omega^{0,\bullet}(C,E)\cong 0$ must vanish, i.e.\ $\tilde{s}=0$.
It then follows that $s_i/\psi_i=\tilde{s}_i=0$, and hence $s_i = 0$,
vanish for all $i\in\mathcal{I}$, implying that $s=0$.
\end{proof}


%
%



\begin{thebibliography}{10}

\bibitem[BCSV26]{BCSV}
M.~Benini, R.~A.~Cullinan, A.~Schenkel and B.~Vicedo,
``On the structure of higher-dimensional integrable field theories,''
{\it in preparation}.

\bibitem[BGMS25]{BGMS}
M.~Benini, A.~Grant-Stuart, G.~Musante and A.~Schenkel,
``Obstructions to the existence of M\o{}ller maps,''
Adv.\ Theor.\ Math.\ Phys.\ {\bf 29}, no.\ 7, 1945--1989 (2025)
[arXiv:2311.00070 [math-ph]].


\bibitem[BSV22]{BSV}
M.~Benini, A.~Schenkel and B.~Vicedo,
``Homotopical analysis of 4d Chern-Simons theory and integrable field theories,''
Commun.\ Math.\ Phys.\ \textbf{389}, 1417--1443 (2022)
[arXiv:2008.01829 [hep-th]].


\bibitem[BL15]{LazarevCS}
C.~Braun and A.~Lazarev,
``Unimodular homotopy algebras and Chern-Simons theory,''
J.\ Pure Appl.\ Algebra \textbf{219}, 5158--5194 (2015)
[arXiv:1309.3219 [math.QA]].


\bibitem[CL24]{CL}
H.~Chen and J.~Liniado,
``Higher gauge theory and integrability,''
Phys.\ Rev.\ D \textbf{110}, no.\ 8, 086017 (2024)
[arXiv:2405.18625 [hep-th]].


\bibitem[CL09]{LazarevHodge}
J.~Chuang and A.~Lazarev,
``Abstract Hodge decomposition and minimal models for cyclic algebras,''
Lett.\ Math.\ Phys.\ \textbf{89}, 33--49 (2009)
[arXiv:0810.2393 [math.QA]].


\bibitem[CG17]{CG1}
K.~Costello and O.~Gwilliam,
{\it Factorization algebras in quantum field theory: Volume 1},
New Mathematical Monographs \textbf{31}, 
Cambridge University Press, Cambridge (2017).


\bibitem[CG21]{CG2}
K.~Costello and O.~Gwilliam,
{\it Factorization algebras in quantum field theory: Volume 2},
New Mathematical Monographs \textbf{41}, 
Cambridge University Press, Cambridge (2021).


\bibitem[CY19]{CY3}
K.~Costello and M.~Yamazaki,
``Gauge theory and integrability, III,''
arXiv:1908.02289 [hep-th].


\bibitem[Cra04]{HPT}
M.~Crainic,
``On the perturbation lemma, and deformations,'' 
arXiv:math.AT/0403266.


\bibitem[DLMV20]{DLMV}
F.~Delduc, S.~Lacroix, M.~Magro and B.~Vicedo,
``A unifying 2D action for integrable $\sigma$-models from 4D Chern-Simons theory,''
Lett.\ Math.\ Phys.\ \textbf{110}, no.\ 7, 1645--1687 (2020)
[arXiv:1909.13824 [hep-th]].


\bibitem[Der21]{Derryberry}
R.~Derryberry,
``Lax formulation for harmonic maps to a moduli of bundles,''
arXiv:2106.09781 [math.AG].


\bibitem[For91]{Forster}
O.~Forster,
{\it Lectures on Riemann surfaces},
Grad.\ Texts in Math.\ \textbf{81},
Springer Verlag, New York (1991).


\bibitem[GRW24]{GRW}
O.~Gwilliam, E.~Rabinovich and B.~R.~Williams,
``Quantization of topological-holomorphic field theories: local aspects,''
Commun.\ Anal.\ Geom.\ \textbf{32}, 1157--1231 (2024)
[arXiv:2107.06734 [math-ph]].


\bibitem[Huy05]{Huybrechts}
D.~Huybrechts,
{\it Complex geometry: An introduction},
Universitext, Springer-Verlag, Berlin (2005).



\bibitem[JRSW19]{JRSW}
B.~Jur\v{c}o, L.~Raspollini, C.~S\"amann and M.~Wolf,
``$L_\infty$-algebras of classical field theories and the Batalin-Vilkovisky formalism,''
Fortsch.\ Phys.\ \textbf{67}, 1900025 (2019)
[arXiv:1809.09899 [hep-th]].


\bibitem[KS24]{KraftSchnitzer}
A.~Kraft and J.~Schnitzer,
``An introduction to $L_\infty$-algebras and their homotopy theory for the working mathematician,''
Rev.\ Math.\ Phys.\ \textbf{36}, 2330006 (2024)
[arXiv:2207.01861 [math.QA]].


\bibitem[Lac22]{Lacroix}
S.~Lacroix,
``Four-dimensional Chern-Simons theory and integrable field theories,''
J.\ Phys.\ A \textbf{55}, 083001 (2022)
[arXiv:2109.14278 [hep-th]].


\bibitem[LV12]{LodayVallette}
J.-L.\ Loday and B.\ Vallette, 
{\it Algebraic operads}, 
Grundlehren der Mathematischen Wissenschaften \textbf{346}, 
Springer Verlag, Heidelberg (2012).


\bibitem[LurX]{Lurie}
J.~Lurie,
{\it Derived algebraic geometry X: Formal moduli problems}.
\url{https://www.math.ias.edu/~lurie/papers/DAG-X.pdf}.


\bibitem[Pri10]{Pridham}
J.~P.~Pridham,
``Unifying derived deformation theories,''
Adv.\ Math.\ \textbf{224}, 772--826 (2010)
[arXiv:0705.0344 [math.AG]].


\bibitem[SV24]{SV}
A.~Schenkel and B.~Vicedo,
``5d 2-Chern-Simons theory and 3d integrable field theories,''
Commun.\ Math.\ Phys.\ \textbf{405}, no.\ 12, 293 (2024)
[arXiv:2405.08083 [hep-th]].


\bibitem[War83]{Warner}
F.~W.~Warner,
{\it Foundations of differentiable manifolds and Lie groups},
Graduate Texts in Mathematics {\bf 94}, 
Springer, Berlin Heidelberg (1983).


\bibitem[Wel07]{Wells}
R.~O.~Wells,
{\it Differential analysis on complex manifolds},
Graduate Texts in Mathematics {\bf 65}, 
Springer, New York (2007).

\end{thebibliography}
\end{document}